\documentclass[twocolumn,secnumarabic,amssymb, nobibnotes, aps, pre]{revtex4-1}
 \usepackage[utf8]{inputenc}
 \usepackage{amsmath}
\usepackage{amssymb}
\usepackage{amsthm}
\usepackage{dcolumn}% Align table columns on decimal point
\usepackage{bm}% bold math
\usepackage{subeqnarray}
\usepackage{xspace}
\usepackage{graphicx}
\usepackage{psfrag}
\usepackage{color}
\usepackage{placeins} % FloatBarrier
\usepackage{empheq}

\newtheorem{remark}{Remark}
\newtheorem{proposition}{Proposition}

\let\wh=\widehat

\newcommand{\dps}{\displaystyle}
\newcommand{\nfrac}[2]{#1/#2}
\newcommand{\texte}[1]{\hbox{\ #1\ }}
\newcommand{\R}{\ensuremath{\mathbb{R}}\xspace}
\newcommand{\T}{\ensuremath{\mathbb{T}}\xspace}
\newcommand{\N}{\ensuremath{\mathbb{N}}\xspace}

\newcommand{\Z}{\ensuremath{\mathbb{Z}}\xspace}

\newcommand{\e}{\ensuremath{\mathrm{e}}}
\newcommand{\ic}{\ensuremath{\mathrm{i}}\xspace} % le i complexe
\newcommand{\dd}{\mathrm{\ d}}
\newcommand\vect[1]{\ensuremath{\mathbf{#1}}}

\newcommand\sign{\mathrm{sign}}
\newcommand{\D}{\ensuremath{\mathcal{D}}\xspace}

\newcommand{\ddz}{\frac{\partial}{\partial t}}
\newcommand{\sfrac}[2]{{\textstyle\frac{#1}{#2}}}
\newcommand\eqdef{\stackrel{\mathrm{def}}{=}}

\newcommand\cF{{\mathcal{F}}}

\newcommand\cI{{\mathcal{I}}}
\renewcommand{\Re}{\operatorname{\rm Re}}
\renewcommand{\Im}{\operatorname{\rm Im}}
\newcommand{\Id}{\mathbb I}
\newcommand\cA{{\mathcal{A}}}

\definecolor{mred}{rgb}{0.7,0.3,0}
\usepackage[colorlinks,citecolor=mred,linkcolor=mred,
            bookmarksopen,
            bookmarksnumbered
           ]{hyperref}

\bibpunct{[}{]}{;}{n}{}{}

\begin{document}

\title{Numerical simulation of Lugiato-Lefever equation for Kerr combs generation in Fabry-Perot resonators}%

\author{Mouhamad Al Sayed Ali}%
\affiliation{UNIV. RENNES, CNRS, IRMAR - UMR 6625, F-35000 Rennes, France}
\author{St\'ephane Balac}%
\email{\textcolor{black}{Corresponding author: stephane.balac@univ-rennes.fr}}
\affiliation{UNIV. RENNES, CNRS, IRMAR - UMR 6625, F-35000 Rennes, France}
\author{Germain Bourcier}%
\affiliation{LAAS-CNRS, Universit\'e de Toulouse, CNRS, UPS, Toulouse, France}
\author{Gabriel Caloz}%
\affiliation{UNIV. RENNES, CNRS, IRMAR - UMR 6625, F-35000 Rennes, France}
\author{Monique Dauge}%
\affiliation{UNIV. RENNES, CNRS, IRMAR - UMR 6625, F-35000 Rennes, France}
\author{Arnaud Fernandez}%
\affiliation{LAAS-CNRS, Universit\'e de Toulouse, CNRS, UPS, Toulouse, France}
\author{Olivier Llopis}%
\affiliation{LAAS-CNRS, Universit\'e de Toulouse, CNRS, UPS, Toulouse, France}
\author{Fabrice Mah\'e}%
\affiliation{UNIV. RENNES, CNRS, IRMAR - UMR 6625, F-35000 Rennes, France}

\begin{abstract}
 Lugiato-Lefever equation (LLE) is a  nonlinear Schr\"odinger
equation with damping, detuning and driving terms,   introduced as a model for  Kerr combs generation in ring-shape resonators and more recently, in the form of a variant, in Fabry-Perot (FP) resonators.
The aim of this paper is to present some  numerical methods  that complement each other  to solve the LLE in its general form   both in the dynamic and  in the steady state regimes. We also provide some mathematical properties of the LLE likely to help the understanding and interpretation of the numerical simulation results.
\end{abstract}
\maketitle
%====================
\section{Introduction}
%====================
Kerr frequency combs refer to a laser source whose spectrum consists of a series of discrete, equally spaced frequency lines,  which are generated in an optical cavity from a pump laser by the Kerr nonlinearity. 
Applications of Kerr frequency combs are numerous in various domains in optics such as optical metrology, lidars or time-frequency systems.
The coherent conversion of the pump laser to a frequency comb has been obtained in a variety of optical resonator such as
  whispering-gallery mode resonators  (WGM),  fiber ring resonators or
   integrated ring-resonators of high non-linearity  \cite{Chembo:16, Chembo:13, Matsko:11, Jang:15}.
 Here we are concerned with  Fabry-Perot (FP) resonators formed by an optical fiber bounded at each end  by  a multi-layer dielectric mirror.
Beyond their ease of design and manipulation \cite{Obrzud:17} FP resonators allow a durable and reproducible coupling in all-fiber optical systems,
enabling their use as a compact source of Kerr frequency combs.
They
are considered as an interesting alternative to crystalline WGM resonators  and  to integrated   optical resonators.
  
\medskip

An approach for the theoretical understanding of
 Kerr comb generation process in WGM   resonators  and  integrated ring-resonators relies on a spatiotemporal partial differential equation referred to as the Lugiato–Lefever equation (LLE), which is a nonlinear Schr\"odinger equation  with damping, detuning, and driving  terms~\cite{Lugiato:87}
 that have been initially introduced in the late 80's by L.~Lugiato and R.~Lefever as a model for spontaneous pattern formation in nonlinear optical systems.
 More recently, this model has been extended to describe  Kerr comb generation process in FP resonators in \cite{Cole:18}.
  The difference between the original LLE for ring shaped (RS) resonators and the equation obtained as a model for Kerr frequency combs generation in a   FP resonator lies on the one hand in the presence of an additional integral term that makes the equation non-local in space and on the other hand on the periodic boundary conditions  resulting from the modeling approach of the  dielectric mirrors at each end of the fiber.
  The LLE for a FP resonator  (in short the FP-LLE) can be expressed as
\begin{subequations}\label{eq:mod}
\begin{align}\label{eq:LLE}
 \frac{\partial \psi}{\partial t} (\theta,t) 
  &=  
 -\ic \frac{\beta}{2} \  \frac{\partial^2 \psi}{\partial \theta^2} (\theta,t) - (1+\ic \alpha)\,\psi(\theta,t) 
\nonumber\\
&\hspace*{-10mm} + \ic  \,   \psi(\theta,t)\, \Big(   |\psi(\theta,t)|^2 +  \frac{\sigma}{\pi} \int_{-\pi}^{\pi} |\psi(\zeta ,t)|^2 \dd \zeta 
  \Big)
+   F
\end{align}
where, in order to  cover in a unified way the two above mentioned  cases  corresponding respectively to RS resonators \cite{Godey:14} and FP resonators \cite{Cole:18},  we  have introduced a normalized space variable $\theta\in\,]-\pi,\pi[$ 
and  the parameter~$\sigma$.
In the case of a  FP resonator with a cavity made from a  fiber with length $L$, the space variable $\theta$ is related to the position $z$ along the fiber by  $\theta = \pi z / L$
whereas for a RS resonator $\theta$ refers to the polar coordinate.
Moreover,  \eqref{eq:LLE} with $\sigma=0$   coincides with the original LLE  for RS whereas for $\sigma=1$
we have the specific LLE for FP resonators.
The constant, unit-less and real valued parameters  in \eqref{eq:LLE} are as follows:
\begin{itemize}
\item[-] $F>0$ is the amplitude of the  continuous wave laser pump;
\item[-] $\beta$ is the group velocity dispersion (GVD)  parameter of the resonator; it can be either positive corresponding to the so-called  \emph{normal regime} or negative     corresponding to the \emph{anomalous regime};
\item[-] $\alpha$ is the cavity phase detuning of the laser pump.
\end{itemize}

\medskip

Equation  \eqref{eq:LLE} is to be considered together with
the following periodic boundary conditions:
\begin{align}\label{eq:BC}
 \psi(-\pi,t) &= \psi(\pi,t) 
\\
  \frac{\partial\psi}{\partial\theta} (-\pi,t) &=   \frac{\partial\psi}{\partial\theta} (\pi,t)  .
\end{align}
Furthermore, an initial condition  is also required:
\begin{equation}\label{eq:IC}
\forall \theta \in[-\pi,\pi]\qquad \psi(\theta,t=0)=\psi_0(\theta)
\end{equation}
\end{subequations}
where $\psi_0$ is a given function that describes the state of the cavity at initial time $t=0$.
Throughout the paper, we will denote by $\| \  \|$ the $\mathbb{L}^2$ norm (a.k.a. energy norm) over $[-\pi,\pi]$ defined  as
\begin{equation}\label{eq:1117}
\|u\| = \Big( \int_{-\pi}^{\pi} |u(\zeta)|^2\dd\zeta \Big)^\frac{1}{2} .
\end{equation}

\medskip

A comprehensive analysis of the way the FP-LLE  can be justified from Maxwell equations
 is conducted in~\cite{Balac:22} and relies on different assumptions
 including  the low transmission approximation, the high-Q limit regime approximation
 and the assumption that  the electric field injected in the FP cavity is   polarized in a direction \vect{e_x}  transverse to the propagation direction $\vect{e_z}$ (see also \cite{Ziani:23}).
 It is established  in~\cite{Balac:22} that,
up to a transformation introduced to handle simultaneously the forward and backward waves through a symmetrization on $[-L,L]$, the FP-LLE solution $\psi$ is related
 to the slowly varying envelop of the  forward $E_F$ and backward $E_B$ electric field wave propagating at a velocity $v_g$ 
 through the relations
 \begin{subequations}\label{eq:1131}
 \begin{align}
E_F(z,t) &= \sum_{n\in\Z} g_n(t) \exp\Big( \ic \frac{\pi n}{L} \big(  v_g t-z) \Big) 
\\
E_B(z,t) &= \sum_{n\in\Z} g_n(t) \exp\Big( \ic \frac{\pi n}{L} \big(  v_g t+z) \Big) 
\end{align}
\end{subequations}
where  the $g_n(t), n\in\Z$, are the Fourier coefficients of~$\psi$ with respect to the space variable.
The electric field component of light-wave in the FP-cavity is expressed  $\forall z\in[0,L], \forall t\in\R$ as
\begin{equation*}
 \vect{ E}(z, t) = \Big( E_F(z,t) \, \e^{\ic \gamma z } +  E_B(z,t) \, \e^{-\ic \gamma z} \Big)\,\e^{-\ic\omega_0 t}\ \vect{e_x} + \mathrm{c.c.}
\end{equation*}
where
$\omega_0$ is the so-called \emph{carrier frequency} of the electric field, $\gamma$ is the mode propagation constant 
and \emph{c.c.} means ``complex conjugate''.

\medskip

Let us clarify the  link between  steady state solutions to the FP-LLE problem~\eqref{eq:mod} and
Kerr frequency combs. (A more physical description in the case of WGM disk resonator can be found  e.g. in \cite{Godey:14}.) 
In the frequency domain, we have $\forall z\in[0,L], \forall \omega\in\R^+$ 
\begin{equation*}
 \vect{\wh E}(z, \omega) = \Big( \wh E_F(z,\omega_0 - \omega) \, \e^{\ic \gamma z } +  \wh E_B(z,\omega_0 - \omega) \, \e^{-\ic \gamma z} \Big)\,\vect{e_x} + \mathrm{c.c.} 
\end{equation*}
where $ \vect{\wh E}(z, \cdot) $ refers to the Fourier transform of $ \vect{E}(z, \cdot) $.
For  steady state solutions  $\psi$  to  FP-LLE, the  sequence $(g_n)_{ n\in\Z}$  is independent of time and it
follows from~\eqref{eq:1131}  that
\begin{align*}
\wh  E_F(z,\omega_0-\omega)   &=  \sum_{n\in\Z} g_n\e^{-\ic \pi n z/L}   \, \delta(\omega_0-\omega+ \pi n \tfrac{v_g}{L} )\\
\wh  
E_B(z,\omega_0-\omega) &= \sum_{n\in\Z}g_n \e^{\ic \pi n z/L} \,  \delta(\omega_0-\omega+ \pi n \tfrac{v_g}{L} ) 
\end{align*}
where $\delta$ is the Dirac delta function.
Thus, in the frequency domain, the electric field related to steady-state solutions $\psi$ to the FP-LLE  at the cavity  end located at $z=L$  is a ``comb'' with equidistant frequencies $\omega$ at a distance $\frac{\pi v_g}{L}$ from each other and centered around the carrier frequency~$\omega_0$.

\medskip

The goals of this paper are to give  an overview of the mathematical  properties of the solutions to the FP-LLE problem in connection with Kerr frequency combs generation and to propose some numerical methods enabling the investigation by numerical simulation of some  features of these solutions,  both in the time-dynamic  regime and in the  steady state regime.
The study of FP-LLE solutions is not a simple problem and may conceal some pitfalls. This is reflected in the numerical simulations in various ways.
\medskip

The paper is organized as follows.
Section  \ref{sec:2} is devoted to some mathematical properties of solutions   to the FP-LLE.
A particular attention is paid to
steady state solutions that are the solutions of interest when studying Kerr frequency combs.
In Section    \ref{sec:3} we  provide some numerical methods for solving the FP-LLE that each gives a different viewpoint when exploring the properties of FP-LLE  solutions.
Finally, in   Section~\ref{sec:4} we  show some results obtained by numerical simulation that illustrate the wealth of situations related to  the FP-LLE.
The paper is completed by two appendices.
Although there is no energy conservation due to the non Hamiltonian structure of the LLE, in Appendix~\ref{app:1124}
an explicit bound of the energy is given.
In  Appendix~\ref{app:1125}, explicit necessary conditions to have bifurcation points from constant solutions are given for both the FP-LLE and its discretization by finite differences.

%==========================================================================
\section{Overview of mathematical properties of the LLE model}\label{sec:2}
%==========================================================================

%= = = = = = = = = = = = = = =
\subsection{Properties of the time-dynamic FP-LLE problem}\label{sec:21}
%= = = = = = = = = = = = = = =
Let us first observe that a solution $\psi$ to  the FP-LLE  problem  \eqref{eq:mod}  can be viewed as a $2\pi$-periodic function defined on the real line~$\R$.
It follows that the function $\theta\in\R \mapsto \psi(\theta - c)$ for any $c\in\R$  provides a solution to   problem  \eqref{eq:mod}.
Thus, periodic boundary conditions imply a translation invariance of solutions. 
Moreover, we can also   point out that periodic solutions to  \eqref{eq:LLE} over \R  exhibit some dilation property with respect to~$\beta$.
Namely, let $\psi$ be a $2\pi$-periodic solution to~\eqref{eq:LLE}  and let $\varphi : \theta \in\R \mapsto \psi(p\theta)$ for a given $p\in\N^*$.
One can easily deduce from~\eqref{eq:LLE}   that the \hbox{$\frac{2\pi}{p}$-periodic} function~$\varphi$  is such that
\begin{align*}
 \frac{\partial \varphi}{\partial t} (\theta,t) 
  &=  
 -\ic \frac{\beta}{2p^2} \  \frac{\partial^2 \varphi}{\partial \theta^2} (\theta,t) - (1+\ic \alpha)\,\varphi(\theta,t) 
\nonumber\\
&\hspace*{-10mm} + \ic  \,   \varphi(\theta,t)\, \Big(   |\varphi(\theta,t)|^2 +  p\,\frac{\sigma}{\pi} \int_{-\pi/p}^{\pi/p} |\varphi(\zeta ,t)|^2 \dd \zeta 
  \Big)
+   F .
\end{align*}
and since $\varphi$ is $\frac{2\pi}{p}$-periodic, we have
\begin{equation*}
 p\, \int_{-\pi/p}^{\pi/p} |\varphi(\zeta ,t)|^2 \dd \zeta 
 = \int_{-\pi}^{\pi} |\varphi(\zeta ,t)|^2 \dd \zeta  .
\end{equation*}
It follows that if  we have a $2\pi$-periodic solution $\psi$ to the FP-LLE for a given value of $\beta$, then  it generates thanks to a dilation transformation with factor $p$, a $\nfrac{2\pi}{p}$-periodic  solution to the FP-LLE for a value of the parameter $\beta$ divided by $p^2$.

\medskip

Let us introduce the  space of $2\pi$-periodic functions {$V=\{v\in\mathsf{H}^1_\mathrm{loc}(\R)\ ;\ v(x-\pi)=v(x+\pi)\,\forall x\in\R\}$}
where $\mathsf{H}^1_\mathrm{loc}(\R)$ denotes the Sobolev space of functions that are square-integrable, as well as their first derivative, over every bounded interval.
In the case  $\sigma=0$, we find in \cite{Ghidaglia:88}
that for every $\psi_0\in V$
the Cauchy problem~\eqref{eq:mod} possesses a unique solution $\psi\in\mathsf{L}^\infty(\R^+,V)$
and that for every $t\in\R^+$  the mapping $\psi_0 \mapsto \psi(\cdot,t)$ is continuous on $V$.
Moreover,  for every $t\in\R^+$, we have the energy estimate
\begin{equation}\label{eq:1124}
\|\psi(\cdot,t)\|^2\leq   \e^{-t} \, \|\psi_0\|^2+  2\pi F^2\,(1-\e^{-t})  .
\end{equation}
The full framework of  \cite{Ghidaglia:88} does not apply to the case $\sigma=1$ due to the  non-linear integral term.
Nevertheless, by the same chain of arguments as in \cite{Ghidaglia:88}, one can show that the energy estimate \eqref{eq:1124} holds as well when $\sigma=1$, see Appendix~\ref{app:1124}.

%= = = = = = = = = = = = = = =
\subsection{Steady state solutions}\label{sec:2108}
%= = = = = = = = = = = = = = =
As mentioned in the introduction, Kerr frequency  combs are related to steady state solutions to the LLE problem~\eqref{eq:mod}.
Such time independent solutions satisfy the following ordinary differential equation on $]-\pi,\pi[$ deduced from  \eqref{eq:LLE} 
\begin{subequations}\label{eq:mods}
\begin{align}\label{eq:LLEs}
 -\ic \frac{\beta}{2}\, \psi''(\theta) 
 &-   (1+\ic \alpha )\,\psi(\theta)  
 \nonumber\\
& +  \ic\  \psi(\theta)\, \big(  |\psi(\theta)|^2 +\frac {\sigma}{\pi}\,\|\psi\|^2 \big)
+ F =0
\end{align}
together with the following  boundary conditions:
\begin{equation}\label{eq:bcs}
\psi(-\pi) = \psi(\pi)\qquad\texte{and}\qquad \psi'(-\pi) = \psi'(\pi) .
\end{equation}
\end{subequations}

The mathematical study of the LLE \eqref{eq:mod} when $\sigma=0$ 
has been the subject of several publications in the last decade.
To our knowledge, the  first mathematical investigation of  the RS-LLE  equation from a viewpoint of bifurcation theory from trivial constant solutions is  \cite{Miyaji:10}. 
A fairly   complete numerical investigation of the various steady-state solutions to the RS-LLE equation has been undertaken in \cite{Godey:14,Godey:17}  where a stability analysis is developed to
obtain  spatial bifurcation maps.
Existence and  stability of periodic solutions to the RS-LLE problem has been investigated in~\cite{ Delcey:18} using a center manifold reduction theory.
This approach provides a detailed  behavior of periodic solutions in a neighborhood
of specific parameter values, but does not allow a global description of solutions when the RS-LLE parameters change.
In  \cite{Mandel:17}  and \cite{Dauge:24} a
 bifurcation analysis for RS-LLE equation
  based on the bifurcation theorems of Crandall and Rabinowitz~\cite{Crandall:71} is conducted.
This approach  provides a different and more global  viewpoint than the one obtained by  the center manifold reduction approach.
 Furthermore, in \cite{Mandel:17} the authors  provide numerical computations of
bifurcation diagrams in various emblematic cases which are very inspiring.

\medskip

\begin{remark}\label{rem:1853}
Note that in the steady state regime, 
given a solution~$\psi$ to the LLE equation \eqref{eq:LLEs} for $\sigma=1$,
the quantity $\alpha_{\psi}=\frac{1}{\pi}\|\psi\|^2$ is a constant and therefore equation~\eqref{eq:LLEs} can be recast in the following way
\begin{equation*}%\label{eq:LLEs2}
 -\ic \frac{\beta}{2}\, \psi''(\theta) 
 -   (1+\ic (\alpha-\alpha_\psi) )\,\psi(\theta) +  \ic\  \psi(\theta)\, |\psi(\theta)|^2 
+ F =0 .
\end{equation*}
The LLE for FP resonators ($\sigma=1$) has the same expression as the LLE for a RS resonator ($\sigma=0$)
with a modified  cavity phase detuning parameter.
\end{remark}
From Remark~\ref{rem:1853} we can expect at first sight to have similar mathematical properties for solutions to the FP-LLE as the ones we have for the RS-LLE that have been extensively studied in the literature.
However, we can not directly deduce the features of solutions to  FP-LLE from known results to the RS-LLE since the change in the detuning parameter depends on the $\mathbb{L}^2$ norm of the solution and therefore is unknown and  solution dependent.
What we  can claim is that each solution $\psi$ to the FP-LLE  for a given  cavity phase detuning parameter   $\alpha$  is also a solution  to the RS-LLE for a  different   detuning parameter given by  $\alpha- \frac{1}{\pi}\|\psi\|^2$.

\medskip

Note  that for a given set of parameters $(\beta,\alpha, F)$,  one can  exhibit up to  three constant solutions (among others steady state solutions). These solutions are often referred to as \emph{flat solutions}, see Section \ref{sec:flat} below. As a consequence, it is clear that  the steady state problem  \eqref{eq:mods}  does not have a unique solution.
\medskip

Finally, let us mention that   problem  \eqref{eq:mods} written for a complex valued unknown $\psi$
can be  set in an equivalent real valued form 
 by considering separately  real and imaginary parts.
 Namely, let $u_1$ denotes the real part of~$\psi$ and $u_2$  its imaginary part; We deduce from~\eqref{eq:LLEs} the following two differential equations satisfied by $(u_1,u_2)$
\begin{subequations}\label{eq:mods2}
\begin{align}
-\frac{\beta}{2} \, u_1'' &- \alpha u_1-  u_2 
\nonumber\\
&+ u_1\,\big(u_1^2+u_2^2+  \frac{\sigma}{\pi}  \, \mathcal{I}(u_1,u_2) \big)  = 0 \label{eq:mods2a}  \\
-\frac{\beta}{2} \, u_2'' &+  u_1-\alpha\, u_2 
\nonumber\\
&+ u_2\,\big(u_1^2+u_2^2 + \frac{\sigma}{\pi} \, \mathcal{I}(u_1,u_2) \big) - F = 0 \label{eq:mods2b} 
\end{align}
where
\begin{equation}
\mathcal{I}(u_1,u_2) = \int_{-\pi}^{\pi} (u_1^2(\zeta) + u_2^2(\zeta) ) \dd\zeta
= \|\psi\|^2 .
\end{equation}
From \eqref{eq:bcs}  we obtain the following periodic boundary conditions for $u_k, k=1,2$
\begin{equation}\label{eq:bcs2}
u_k(-\pi) = u_k(\pi)\qquad\texte{and}\qquad u_k'(-\pi) = u_k'(\pi) .
\end{equation}
\end{subequations}

%= = = = = = = = = = = = = = =
\subsection{Flat solutions}\label{sec:flat} 
%= = = = = = = = = = = = = = =

A special kind of steady state solutions to the  FP-LLE~\eqref{eq:LLE} are  
solutions that depend neither on time nor on position.
These trivial solutions  are called \emph{flat} solutions and they  can be explicitly calculated as detailed below.
From a practical viewpoint, their interest in the study of steady state solutions to the LLE is related to the practical observation, see Section~\ref{sec:4}, that 
stable spatially periodic solutions bifurcate from this set of solutions when the LLE parameters $\alpha$ or $F$ vary. 
Moreover,  these flat solutions behave as attractors when using numerical methods such as the Split-Step method or a Collocation method: Namely,
a special attention has to be paid to the choice of the computational setting in order to get a solution other than a flat solution.

For a flat solution $\psi_\bullet$, the FP-LLE \eqref{eq:LLE}  reads
\begin{equation}\label{eq:LLEf}
   (1+\ic \alpha)\,\psi_\bullet  - \ic  \,   \psi_\bullet\, \Big(   |\psi_\bullet|^2 + 2\sigma |\psi_\bullet|^2  \Big) = F .
\end{equation}
By taking  into account the product of   \eqref{eq:LLEf} by its conjugate, we deduce that
\begin{equation}\label{eq:1518}
\big( 1 + (\alpha - (1+2\sigma)\rho_\bullet)^2\big)\,\rho_\bullet  = F^2
\end{equation}
where we have set
\begin{equation}\label{eq:1520}
\rho_\bullet = |\psi_\bullet|^2 .
\end{equation}
One may  wonder how many flat solutions   the FP-LLE~\eqref{eq:LLE} holds.
To answer the question we must determine how many  \emph{positive} solutions $\rho_\bullet$ equation \eqref{eq:1518} owns for a couple of parameters $(\alpha, F)$ fixed.
Since $\rho_\bullet>0$, we are interested in the positive roots of the  polynomial $P\in\R[X]$ where
 \begin{equation}\label{eq:1638}
 P \eqdef
(1+2\sigma)^2\, X^3 -2\alpha(1+2\sigma)\, X^2+ (1+\alpha^2) X - F^2 .
\end{equation}
As well known,  $P$ may have    one real root or three real roots (counted with multiplicity); However,   the number of positive roots may  \emph{a priori}  vary from zero to three. 
Note that once $\rho_\bullet$ is known as one of the root of $P$,
one can easily deduce the corresponding flat solution $\psi_\bullet$.
From equation~\eqref{eq:LLEf}, we have
 \begin{equation}
\psi_\bullet = \frac{F}{1+\ic (\alpha-(1+2\sigma)\rho_\bullet)}. 
\end{equation}

\medskip

To investigate existence of positive roots to  $P$,  let $G_\alpha : \R\rightarrow \R$ be the mapping defined for a fixed value of $\alpha$ considered as a parameter by
\begin{align}
G_\alpha(\rho) &=  \big( 1 + (\alpha - (1+2\sigma)\rho)^2\big)\,\rho .
\end{align}
 $G_\alpha$ is a degree 3 polynomial mapping with $G_\alpha(0)=0$ and 
$\dps\lim_{\rho\rightarrow\pm\infty} G_\alpha(\rho) = \pm \infty$.
Its derivative is given by
\begin{equation}\label{eq:1536}
G_\alpha'(\rho) =  3(1+2\sigma)^2\, \rho^2 -4\alpha(1+2\sigma)\, \rho+ (1+\alpha^2)
\end{equation}
and the discriminant of the binomial equation $G'_\alpha(\rho) = 0$  is
$
\Delta = 4(1+2\sigma)^2\, (\alpha^2-3) .
$
Thus, equation $G_\alpha'(\rho)=0$ has two conjugate complex roots when $\alpha^2<3$, one  real root when $\alpha^2=3$ and two real roots     when $\alpha^2>3$.
It follows that $G_\alpha$ is a strictly increasing function over \R when $\alpha^2\leq 3$.
Since $G_\alpha(0)=0$, it defines a one-to-one mapping from $[0,+\infty[$ onto $[0,+\infty[$  and we can conclude that when $\alpha^2\leq 3$ equation \eqref{eq:1518} has a unique positive solution.

The case when $\alpha^2>3$ is less straightforward.
Solving equation $G_\alpha'(\rho) =0$, we find that $G_\alpha$ has two local extrema located at
\begin{equation}\label{eq:1556}
\rho_\pm = \frac{2\alpha \pm \sqrt{\alpha^2-3}}{3(1+2\sigma)} .
\end{equation}
Moreover, one can easily check that when $\alpha$ is positive, the two  local extrema  abscissas $\rho_\pm$
 are positive with \hbox{$0< \rho_-< \rho_+$} and that when $\alpha$ is negative
 $\rho_-< \rho_+ < 0$.
Thus, when  $\alpha$ is negative, $G_\alpha$ is strictly increasing over $[0,+\infty[$.
Since $G_\alpha(0)=0$, it defines a one-to-one mapping from $[0,+\infty[$ onto $[0,+\infty[$  and we  conclude that when $\alpha< -\sqrt {3}$ equation~\eqref{eq:1518} has a unique positive solution.

Finally, it remains to consider the case when $\alpha>\sqrt{3}$ where $G_\alpha$ increases from $0$ to $\rho_-$, then decreases from $\rho_-$ to $\rho_+$ and finally increases from   $\rho_+$  to infinity.
We have
\begin{align*}
G_\alpha(\rho_+) &= \frac{2\,\alpha+\sqrt {{\alpha}^{2}-3}}{27 (1+2\sigma) } 
 \left( 9+ \left( \alpha-\sqrt {{\alpha}^{2}-3} \right) ^{2}
 \right) \\
 G_\alpha(\rho_-) &=\frac { 2\,\alpha-\sqrt {{\alpha}^{2}-3}}{ 27(1+2\sigma) }
 \left(9+ \left( \alpha+\sqrt {{\alpha}^{2}-3} \right) ^{2}
 \right)
\end{align*}
and for $\alpha>\sqrt{3}$ these two quantities are positives.
It follows that
\begin{itemize}
\item[-] when $F^2 \in] G_\alpha(\rho_+), G_\alpha(\rho_-)[$, equation~\eqref{eq:1518} has a three positive solutions;
\item[-] when $F^2 = G_\alpha(\rho_+)$ or  when $F^2 =  G_\alpha(\rho_-)$, equation~\eqref{eq:1518} has a two positive solutions;
\item[-] when $F^2 < G_\alpha(\rho_+)$ or when $F^2 >  G_\alpha(\rho_-)$, equation~\eqref{eq:1518} has a  unique positive solution.
\end{itemize}

Note that when there exist three positive solutions $\rho_1, \rho_2, \rho_3$  to  equation~\eqref{eq:1518}, they can be arranged in increasing order as follows
\begin{equation}
0 < \rho_1 < \rho_- < \rho_2 < \rho_+ < \rho_3 .
\end{equation}

We have depicted in Fig.~\ref{fig:1639} the plane $(\alpha,F^2)$ and we have indicated the different areas corresponding to the number of solutions to equation~\eqref{eq:1518} when $\sigma=1$.
One can note that  flat solutions do not depend on the value of the parameter $\beta$.
We have also depicted in Fig.~\ref{fig:1640}  the variation of $\rho_\bullet = |\psi_\bullet|^2$ as a function of $F$ for $\alpha\in\{1,5,10,15\}$.
For a fixed value of $F$ and $\alpha\in\{5,10,15\}$, one can see that one may have $1$, $2$ or $3$ corresponding values for $\rho_\bullet$ (intersection of the curves with the vertical line
passing through $(F,0)$ whereas for $\alpha=1$, we have only one values for $\rho_\bullet$ whatever is the value of $F$.
A similar observation can be done  in Fig.~\ref{fig:1641} where we have depicted    the variation of $\rho_\bullet = |\psi_\bullet|^2$ as a function of $\alpha$ for  various values of $F$.

\begin{figure}[h]
\begin{center}
\includegraphics[width=\linewidth]{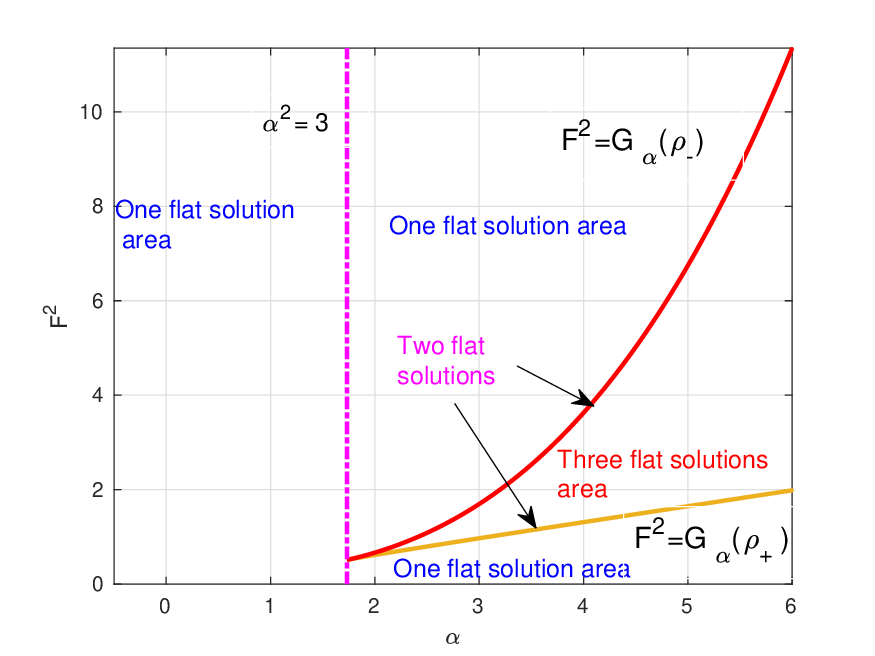}
\caption{Number of flat solutions to the LLE  depending on the values of the parameters $(\alpha, F^2)$ for $\sigma=1$.}\label{fig:1639}
\end{center}
\end{figure}

\begin{figure}[h]
\begin{center}
\includegraphics[width=0.75\linewidth]{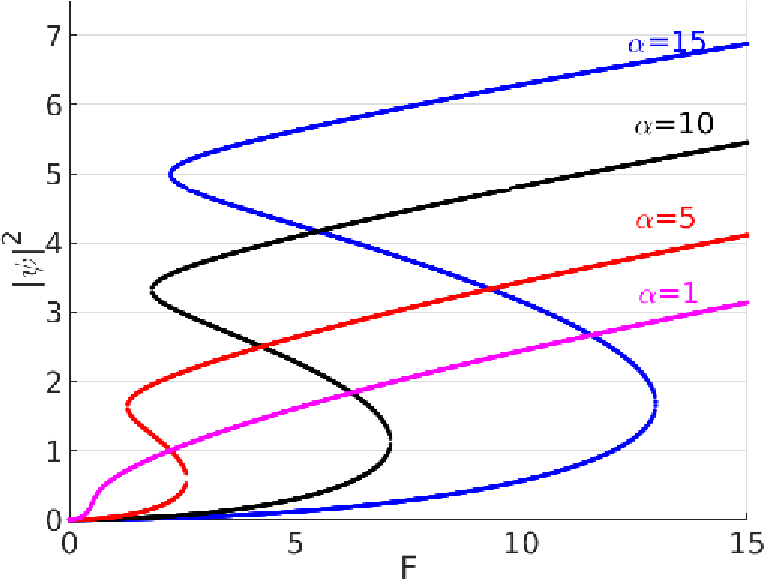}
\caption{Variation of $\rho_\bullet = |\psi_\bullet|^2$ as a function of $F$ for $\alpha\in\{1,5,10,15\}$ and $\sigma=1$.}\label{fig:1640}
\end{center}
\end{figure}

\begin{figure}[h]
\begin{center}
\includegraphics[width=0.75\linewidth]{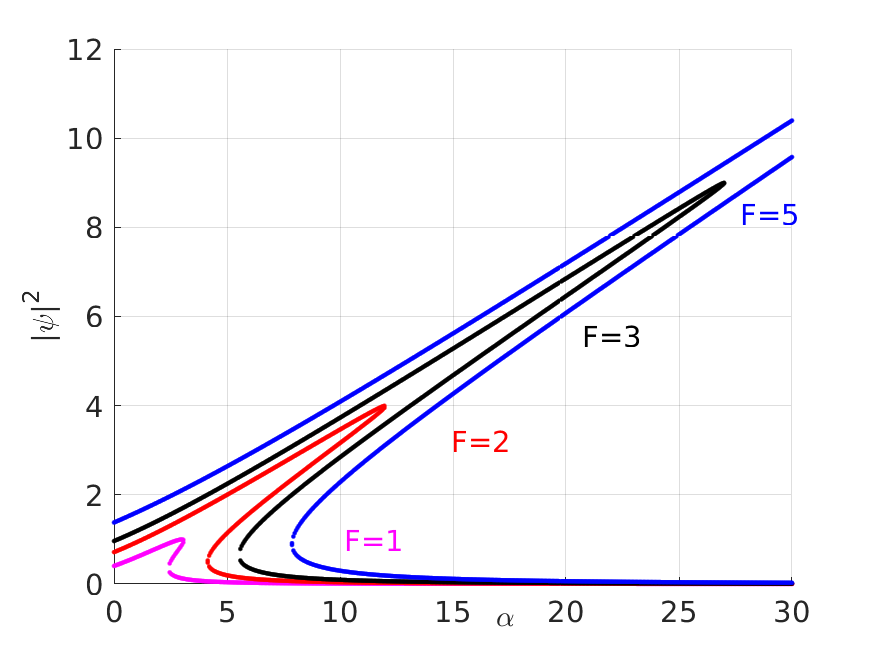}
\caption{Variation of $\rho_\bullet = |\psi_\bullet|^2$ as a function of $\alpha$ for $F\in\{1,2,3,5\}$ and $\sigma=1$.}\label{fig:1641}
\end{center}
\end{figure}

\medskip

As a last remark, let us mention that when the non-linear part of the LLE is handled through  a linearization process, one can easily show that the resulting linear ODE problem has a \emph{unique} solution that is a flat solution.
This     helps understanding the special role of flat solutions for the LLE.

%= = = = = = = = = = = = = = =
\subsection{Properties of FP-LLE steady state solutions}
%= = = = = = = = = = = = = = =
We have  several noticeable results for solutions to the steady state FP-LLE~\eqref{eq:mods} considered here in its equivalent form~\eqref{eq:mods2}
that extend the results of \cite{Mandel:17} for RS-LLE.

\begin{proposition}\label{prop:2218}
Any solution $\psi = u_1+\ic u_2$  to the  steady state FP-LLE~\eqref{eq:mods} 
satisfies the energy estimate
\begin{equation*}
 \|\psi\|\leq  \sqrt{2\pi} {F} .
\end{equation*}
Moreover, $ \int_{-\pi}^{\pi}  u_1( \theta)\dd\theta >0$.
\end{proposition}
\begin{proof}
This is  a direct consequence of the energy estimate~\eqref{eq:1124} for  steady state solutions to the FP-LLE.
The statement $ \int_{-\pi}^{\pi}  u_1( \theta)\dd\theta >0$ is a direct consequence of  relation
\eqref{eq:1232} in Appendix~\ref{app:1124}.
\end{proof}

\begin{proposition}\label{prop:2219}
Any solution $\psi$
 to the  steady state LLE~\eqref{eq:mods} satisfies
\begin{equation*}
\|\psi\|_{\infty} \eqdef
\sup_{\theta\in\T} |\psi(\theta)|\leq   F  + 
 \frac{24\,\pi^{2}}{|\beta|}\,  F^3 .
\end{equation*}
\end{proposition}
\begin{proof}
For the RS-LLE ($\sigma=0$), this estimate is proved in
\cite[Section~3]{Mandel:17}.
It does not rely on the detuning parameter.
So, it remains valid for any value of this parameter.
As a consequence, the estimate is also true for the FP-LLE  ($\sigma=1$),  see Remark~\ref{rem:1853}.
\end{proof}

An important feature of  the   steady state LLE~\eqref{eq:mods} is that for $\beta$ and $F$ fixed, outside a given range of values for $\alpha$,  the only solutions are   flat solutions as stated in Proposition~\ref{prop:1443} below.

\begin{proposition}\label{prop:1443}
There exist $\alpha_\sigma^\star > 0$ and $\alpha_{\sigma \star} < 0$, only dependent on the parameters $F$ and $\beta$, such that if $\sign(-\beta)\ \alpha \notin  [\alpha_{\sigma \star},\alpha_{\sigma}^{\star}]$ then  any steady state solution to the LLE  \eqref{eq:mods} 
is a flat solution.
\end{proposition}
\begin{proof}
For the RS-LLE ($\sigma=0$), this result is obtained directly from  \cite[Theorem 1.2]{Mandel:17}.
We can deduce the property for the steady state solutions $\psi$ to the FP-LLE from the RS-LLE case as follows.
As stated in Remark~\ref{rem:1853}, the FP-LLE can be recast into a RS-LLE with  a detuning parameter
$\alpha_{RS} = \alpha - \frac{1}{\pi}\|\psi\|^2$ and therefore we can state that 
if $\sign(-\beta)\,\big(\alpha - \frac{1}{\pi}\|\psi\|^2 \big) 
\notin  [\alpha_{0 \star},\alpha_0^\star]$ then $\psi$
is a flat solution.
Thanks to the energy estimate of Proposition~\ref{prop:2218}, we deduce that
\begin{itemize}
\item[-] if $\beta<0$ and for $\alpha< \alpha_{0\star}$ or $\alpha>\alpha_0^\star+2F^2$ 
\item[-] if $\beta>0$ and for $\alpha< -\alpha_{0 \star}$ or $\alpha> 2F^2-\alpha_0^\star$  
\end{itemize}
any stationary solution to the FP-LLE 
is a flat solution.
We can merge the two cases as : if $\sign(-\beta)\ \alpha \notin  [\alpha_ {1 \star},\alpha^\star_1]$ then  any stationary solution to the FP-LLE is a flat solution,
where we have set
$\alpha_ {1 \star} = \alpha_{0\star}$ and $\alpha_1^\star = \alpha_0^\star+2F^2$ if $\beta<0$
and  
$\alpha_ {1 \star} = \alpha_{0 \star} - 2F^2$ and $\alpha_1^\star = \alpha_0^\star$ if $\beta>0$.
\end{proof}

Note that explicit bounds for quantities  $\alpha_{0 \star}$ and $\alpha_0^\star$ can be found in  \cite{Mandel:17}.

%====================
\section{Some numerical methods for solving   FP-LLE}\label{sec:3}
%====================

In this section, we describe three complementary  numerical methods.
We start by presenting   a Split-Step method to solve the dynamic FP-LLE  \eqref{eq:mod},
then we present  a Collocation method aimed at solving the steady state FP-LLE~\eqref{eq:mods}.
Finally, in order to uncover in a more systematic way branches of non constant steady state solutions, we present a continuation method  according to the parameter $F$ or $\alpha$.

%= = = = = = = = = = = = = = = = = = = = = = = =
\subsection{Symmetric Split-Step  method}\label{sec:S3F4LLE}
%= = = = = = = = = = = = = = = = = = = = = = = =
An elementary idea   to investigate Kerr frequency combs  consists in solving the time-dynamic FP-LLE~\eqref{eq:mod}
and observing the solution on a long time scale.
Here we choose the Split-Step method.
The principle of  the  Split-Step method is to introduce a subdivision of the time interval and over each sub-interval to solve in a prescribed order the linear and the non-linear parts of the FP-LLE~\eqref{eq:LLE}, each of the resulting sub-problem being simpler to solve.
Thus, for $t>0$ fixed, let us consider the linear operator
\begin{equation*}
%  \label{eq:2211}
  \mathcal D : \psi(\cdot,t)  \longmapsto
 -\ic \frac{\beta}{2} \  \frac{\partial^2 \psi}{\partial \theta^2} (\cdot,t) - (1+\ic \alpha)\,\psi(\cdot,t)   +   F
\end{equation*}
 and the non-linear operator
\begin{equation*} 
% \label{eq:2212}
  \mathcal N : \psi(\cdot,t)  \longmapsto \ic    \,   \psi(\cdot,t)\, \Big(   |\psi(\cdot,t)|^2 +   \frac{\sigma}{\pi} \int_{-\pi}^{\pi} |\psi(\zeta ,t)|^2 \dd \zeta   \Big) 
\end{equation*}
so that with these notations, the FP-LLE~\eqref{eq:LLE} reads
$
     \ddz \psi(\cdot,t) =  \mathcal D\, \psi(\cdot,t) + \mathcal N(\psi(\cdot,t))   .
$
Moreover,
let us consider a subdivision  $(t_k)_{k \in\{ 0, \ldots,K\}}$  of 
the time interval  $[0,T]$ and let $t_{k+\frac{1}{2}} = t_k + \frac{h_k}{2}$
where $h_k = t_{k+1}-t_k$ is  the current  step-size.
The Symmetric Split-Step method consists in solving over each sub-interval $[t_k,t_{k+1}]$, the following three nested problems
with space variable $\theta$ as a parameter:
\begin{subequations}
\begin{equation}
  \label{eq:9418k}
 \left\{
   \begin{aligned}
      \ddz u_k(t) &=  \mathcal N(u_k(t)) \qquad \forall t\in[t_k,t_{k+\frac{1}{2}}] \\
 u_k(t_k) &= \psi_{k} 
   \end{aligned}\right.
\end{equation}
where for $k=0$, $\psi_0$ is the initial data and for $k\geq 1$,
$\psi_{k}$ represents the approximated solution at grid point $t_k$ computed at the previous step;
\begin{equation}
  \label{eq:9419k}
  \left\{
   \begin{aligned}
      \ddz v_k(t) &=    \D \, v_k(t) \qquad \forall t\in[t_k,t_{k+1}] \\
 v_k(t_k) &= u_k(t_{k+\frac{1}{2}})  
   \end{aligned}\right.
\end{equation}
where $u_k(t_{k+\frac{1}{2}})$ represents the solution to problem \eqref{eq:9418k}
at half grid point $t_{k+\frac{1}{2}}$;
\begin{equation}
  \label{eq:9420k}
 \left\{
   \begin{aligned}
      \ddz w_k(t) &=    \mathcal N(w_k(t))  \qquad \forall t\in[t_{k+\frac{1}{2}},t_{k+1}] \\
 w_k(t_{k+\frac{1}{2}}) &= v_k(t_{k+1}) 
   \end{aligned}\right.
\end{equation}
\end{subequations}
where $v_k(t_{k+1})$ represents the solution to problem~\eqref{eq:9419k}
at node $t_{k+1}$.
An approximated solution to the FP-LLE~\eqref{eq:LLE} at grid node~$t_{k+1}$ is then given by
\[
\psi(t_{k+1})\approx   w_k(t_{k+1}) \eqdef \psi_{k+1} .
\]
The efficiency of the Symmetric Split-Step method to solve the FP-LLE  \eqref{eq:mod}
relies on the fact that both the linear and non-linear problems in the above splitting can be solved easily. 
Namely,  on the one hand, the  solution to problem~\eqref{eq:9419k} can be computed by using the Fourier approach and we get the following explicit formula
\begin{equation}
v_k(t_{k+1}) =\frac{ \e^{h_k(1+\ic\alpha)} -1 }{1+\ic\alpha} \, F+ \sum_{n\in\Z} \mu_n \, \e^{d_n h_k}\ \e^{\ic n\theta}  
\end{equation}
where the $\mu_n$ are the Fourier coefficients of the $2\pi$-periodic function
$u_k(t_{k+\frac{1}{2}})$ 
and  $d_n = \ic\, \frac{\beta}{2} n^2 - (1+\ic\alpha)$.
On the other hand, 
the solution to problem~\eqref{eq:9418k}   can be computed analytically
from  the following integral representation form:
\begin{equation}\label{eq:0954}
u_k(t) = \psi_k\ \exp\left(   \ic \int_{t_k}^t  |u_k(\tau)|^2
+ \sfrac{\sigma}{\pi}  \|u_k(\tau)\|^2_0
 \dd\tau \right) .
\end{equation}
By multiplying each side of the ODE in \eqref{eq:9418k} by $\overline{u_k}(t)$
and adding it to the complex conjugate equation deduced from~\eqref{eq:9418k}, which has been previously multiplied by ${u_k}(t)$, we can
show that $ \frac{\partial}{\partial t} |u_k(t)|^2   = 0$.
As a consequence $|u_k(t)|^2$ does not depend on $t$
and therefore we have the following explicit expression for the solution deduced from \eqref{eq:0954}:
\begin{equation*}
u_k(t) = \psi_k\ \exp\Big(  \ic\, (t-t_k) \big( |\psi_k|^2
+ \sfrac{\sigma}{\pi}  \|\psi_k\|^2_0\big)\Big) .
\end{equation*}
A similar expression holds for the solution to \eqref{eq:9420k} with $\psi_k$ changed for $v_k(t_{k+1}) $.

We have developed a free open-source program under \textsc{Matlab} that solves the  time-dynamic FP-LLE~\eqref{eq:mod} by the Split-Step method described here, see \cite{Balac:22a}.

The main drawbacks of this numerical approach is that depending on the initial  data $\psi_0$  in \eqref{eq:IC} it
may or may not exist a steady state solution  to problem \eqref{eq:mod}  and even when the solution is stationary for the choice made for $\psi_0$,
it is not obvious  to be sure that a steady state has been reached at the end of the simulation.
Moreover, as usually for such an explicit numerical scheme for a wave-type equation a numerical stability condition (CFL  like condition)  that links the time step-size to the space step-size exists and imposes to carefully chose the time step-size according to the space step-size \cite{Jahnke:17}.
However, this simulation approach is very efficient in terms of computational cost and our working experience has shown that with  appropriately manufactured initial data $\psi_0$, it is possible to obtain a large variety of steady state solutions to the FP-LLE problem.

%= = = = = = = = = = = = = = = = = = = = = = = = = = =
\subsection{A Collocation method}\label{sec:1640}
%= = = = = = = = = = = = = = = = = = = = = = = = = = =

Another way to compute an approximate solution to the steady-state FP-LLE~\eqref{eq:mods} is to use a collocation approach where basically a subdivision of the interval $[-\pi,\pi]$ is introduced leading to a non-linear system  of algebraic equations resulting from the boundary conditions and the collocation conditions imposed at the internal nodes of the subdivision. 
The collocation conditions can be obtained in various way from  ODE  \eqref{eq:LLEs}.
An  elementary way  is to use a finite difference approximation  at the nodes of the
 subdivision $(\theta_n)_{n=0,\ldots,N}$
with constant step-size $h = 2\pi/N$ such that $\theta_n= -\pi + n\,h $.
At each node $\theta_n$ of the subdivision, we  approximate the second order derivative $u_k''(\theta_n)$, $k=1,2$, in  \eqref{eq:mods2a} -- \eqref{eq:mods2b}    by a centered finite difference formula:
\begin{equation}
u_k''(\theta_n) \approx \frac{u_{k}(\theta_{n+1})-2u_{k}(\theta_{n})+u_{k}(\theta_{n-1})}{h^2}
\end{equation}
As well known this approximation formula is second order accurate with respect to the step-size~$h$.
Moreover, we can approximate the integral $\mathcal{I}(u_1,u_2)$ from the values of  $u_k$ at the subdivision nodes by using the trapezoidal rule that reads, due to  the periodic boundary conditions,
\begin{equation}
\mathcal{I}(u_1,u_2) \approx 
 h \sum_{n=1}^{N} \Big( u_1^2(\theta_n)+u_2^2(\theta_n) \Big) .
\end{equation}
It is well known that the trapezoidal quadrature formula  has super-convergence properties when applied to the computation of the integral of a smooth periodic function over a period \cite{Trefethen:14}.

Denoting by $u_{k,n}$, $n=0,\ldots, N$, the approximate value of $u_k(\theta_n)$,
we are led to the set of non-linear equations
for $ n=1,\ldots, N $
{\small
\begin{subequations}\label{eq:2202}
\begin{align}
&
-\frac{1}{2}\beta \,  \frac{u_{1,n+1}-2u_{1,n}+u_{1,n-1}}{h^2}- \alpha u_{1,n}- u_{2,n} 
\nonumber\\
&+ u_{1,n}\,\big(u_{1,n}^2+u_{2,n}^2\big)
+\frac{\sigma\, h}{\pi} \,  u_{1,n} \sum_{p=1}^{N} \big( u_{1,p}^2+u_{2,p}^2 \big)
=0
 \label{eq:2202a}
 \\
&-\frac{1}{2}\beta \,  \frac{u_{2,n+1}-2u_{2,n}+u_{2,n-1}}{h^2} - \alpha\, u_{2,n} +u_{1,n}  \nonumber\\
&
+ u_{2,n}\,\big(u_{1,n}^2+u_{2,n}^2 \big) 
+\frac{\sigma\, h}{\pi} \,  u_{2,n} \sum_{p=1}^{N} \big( u_{1,p}^2+u_{2,p}^2 \big)
 = F
 \label{eq:2202b}
 \end{align}
\end{subequations}
}
Additionally,  we have the relation
\begin{equation}\label{eq:0858}
u_{k,N} = u_{k,0}\qquad\forall k=1,2
\end{equation}
deduced from the  periodic boundary condition  \eqref{eq:bcs2}.
Denoting by $U = \big( u_{1,1},\ldots, u_{1,N}, u_{2,1},\ldots, u_{2,N}\big)^\top \in\R^{2N}$ the vector of unknowns,  we deduce from \eqref{eq:2202}--\eqref{eq:0858}
 the following non-linear system of $2N$ equations
\begin{equation}\label{eq:1218}
\mathrm{M} \, U + \mathcal{N}(U) = 0
\end{equation}
where $\mathrm{M} \in\mathcal{M}_{2N}(\R)$ is the block matrix given by
\begin{equation}\label{eq:1010}
\mathrm{M}  = \begin{pmatrix}
\frac{\beta}{2h^2} \,\mathrm{A} - \alpha\,\mathrm{I}_N & -\mathrm{I}_N \\[2mm]
 \mathrm{I}_N & \frac{\beta}{2h^2} \,\mathrm{A} -  \alpha\,\mathrm{I}_N
\end{pmatrix}
\end{equation}
 where $\mathrm{I}_N$ refers to the identity matrix in  $\mathcal{M}_{N}(\R)$
 and $\mathrm{A}$ is the matrix in  $\mathcal{M}_{N}(\R)$ defined as
 \begin{equation}\label{eq:1025}
\mathrm{A}  = \begin{pmatrix}
2 & -1 & 0 & \cdots & 0 & -1\\
-1 & 2 & -1 &0&\cdots&0\\
0&\ddots& \ddots & \ddots&&\vdots\\
\vdots&&\ddots& \ddots &\ddots&0 \\
0&&&-1 & 2 & -1 \\
-1 & 0 &\cdots &0 &-1 & 2 
\end{pmatrix} .
\end{equation}
The non-linear term is defined as
{\small
\begin{equation}
\mathcal N(U) = 
\begin{pmatrix}
u_{1,1}\,\big(u_{1,1}^2+u_{2,1}^2  \big) \\
\vdots \\
u_{1,N}\,\big(u_{1,N}^2+u_{2,N}^2 \big) \\
u_{2,1}\,\big(u_{1,1}^2+u_{2,1}^2  \big) \\
\vdots \\
u_{2,N}\,\big(u_{1,N}^2+u_{2,N}^2 \big) \\
\end{pmatrix}
+ \frac{\sigma\, h}{\pi} \|U\|^2  \ U 
- \begin{pmatrix}
0 \\ \vdots \\ 0 \\ F \\ \vdots \\ F  
\end{pmatrix}
\end{equation}
}%
where $\|U\|^2 = U^\top \ U$.
Finally, we obtain the approximate solution to the LLE-FP \eqref{eq:mods2} by solving the non-linear system \eqref{eq:1218} by a Newton solver. 

\medskip

The difficulty in using a Collocation method for solving the   FP-LLE~\eqref{eq:mods}
comes from the non-uniqueness of its solutions and the high dependence on the initial guess 
 of Newton iterative algorithm on  the convergence properties  and on the computed solution.
 In practice, our working experience shows that it is not at all easy to find an initial guess leading to the computation of a steady state solution different from a flat solution. Nevertheless, the method can be very efficient when  
  some prior information is available to design  a suitable initial guess.
 \medskip
 
 Note that these difficulties are not specific to the Collocation method presented here.
We have also developed  a free open-source  \textsc{Matlab} toolbox that uses
  the alternative Collocation method implemented in \textsc{Matlab} solver \texttt{bvp4c}, see \cite{Balac:23}. This Collocation method  is based on   a three-stage Lobatto  formula \cite{Shampine:01} and  the same phenomena were observed.

%= = = = = = = = = = = = = = = = = = = = = = = =
\subsection{Pseudo-arclength continuation method}
%= = = = = = = = = = = = = = = = = = = = = = = =

The above mentioned numerical issues and the wish to have a global picture of steady state solutions
 led us as in \cite{Mandel:17} to investigate Kerr frequency combs generation through   a pseudo-arclength continuation method \cite{Allgower}.
In the sequel we will denote by $\lambda$ one of the two parameters $F$ or $\alpha$ chosen as continuation parameter.
For ease of exposition,
we consider the discrete version of the steady state FP-LLE obtained in Section~\ref{sec:1640}  expressed in the form of the non-linear system 
$ \mathrm{M}\,U +  \mathcal N(U) =0$,
see \eqref {eq:1218}.
Note that a   development similar to the one presented  here can be done   with  equations  \eqref{eq:mods2a} -- \eqref{eq:mods2b} discretized by the Finite Element Method.

\medskip

We introduce  the mapping 
\begin{equation}\label{eq:1540}
G_h : (U,\lambda) \in \R^{2N} \times \R \longmapsto \mathrm{M}\,U +  \mathcal N(U) \in \R^{2N} 
\end{equation}
so that  solving the non-linear system   \eqref {eq:1218}
amounts to look for $(U,\lambda)$ such that $G_h (U,\lambda) =0$.
As stated before and illustrated in Section~\ref{sec:4},   solving the non-linear system   \eqref {eq:1218} by a Newton like solver is very dependent on the initial guess used for Newton iterations.
To overcome this difficulty, one can use a continuation method.
The idea  
behind the
\emph{natural parameter continuation}  method
 is to start from
a known solution  $(U^{[0]}, \lambda^{[0]}) $
and to compute a solution to the non-linear  system   \eqref {eq:1218} for a parameter value $\lambda^{[1]} = \lambda^{[0]} + \delta\lambda$ for a small increment $ \delta\lambda$ 
by Newton method with initial guess  $U^{[0]}$.
The solution  $U^{[1]}$ provided by Newton method for the parameter value $\lambda^{[1]}$ is then used as an initial guess for Newton method applied to the solving of  system   \eqref {eq:1218}  for the new  parameter value $\lambda^{[2]} = \lambda^{[1]} + \delta\lambda$.
And this process is continued by incrementing the value of $\lambda$  step by step. 
Such an approach is justified by the \emph{Implicit Function Theorem} (IFT).
According to the IFT  if $(U^{[0]}, \lambda^{[0]})$ is a solution to
the non-linear system   \eqref {eq:1218} and $\partial_U G_h(U^{[0]}, \lambda^{[0]})$ is not singular, then 
there exists a unique mapping $g$ from  $] \lambda^{[0]}-\varepsilon, \lambda^{[0]}+\varepsilon[$ onto  the open ball $B(U^{[0]},r)$   such that
the set of solutions to  equation $G_h(U, \lambda) = 0$  near $(U^{[0]}, \lambda^{[0]})$  is a $\mathcal{C}^1$ curve given by 
$\{ (\lambda, g(\lambda)) \ ;\ \lambda \in ] \lambda^{[0]}-\varepsilon, \lambda^{[0]}+\varepsilon[ \}$
with $U^{[0]} = g(\lambda^{[0]})$.

\medskip

We denote by $\partial_U G_h$ the Jacobian of $G_h$ with derivatives taken  with respect to the variable $U$ only.
A  simple  calculation   shows  that 
{\small
\begin{equation}\label{eq:1639}
\partial_U {G_h}(U,\lambda)= \mathrm{M} + \mathrm{J}(U)  + \frac{\sigma\, h}{\pi} \,\big( \big( U^\top U \big)\ \mathrm{I}_{2N} + 2 \, U\, U^\top\big)
\end{equation} 
}
where $\mathrm{J}(U) $ is the  symmetric block-matrix
\begin{equation*}
\mathrm{J}(U)  = \begin{pmatrix}
J_{11} & J_{12} \\
J_{12} & J_{22}
\end{pmatrix}
\in\mathcal{M}_{2N}(\R)
\end{equation*}
and $J_{11}$, $J_{21}$,  $J_{22}$ are three diagonal matrices in $\mathcal{M}_{N}(\R)$
with entries
\begin{itemize}
\item $3u_{1,n}^2 +  u_{2,n}^2$, \ $n=1,\ldots, N$ for $J_{11}$
\item $ 2 u_{1,n}u_{2,n}$, \ $n=1,\ldots, N$ for $J_{12}$
\item $u_{1,n}^2 + 3 u_{2,n}^2$, \ $n=1,\ldots, N$ for $J_{22}$.
\end{itemize}

\medskip

For the study of the FP-LLE, the idea is to consider for $(U^{[0]}, \lambda^{[0]})$ a flat solution since they are explicitly known, see Section~\ref{sec:flat}.
By the continuation method, we can construct from step to step the curve of flat solutions when the parameter $\lambda$ increases. This has no interest (since flat solutions are explicitly known and do not require a numerical computation) until we reach a point $(U,\lambda)$ where  the assumptions of  the IFT are violated, in
particular when $\partial_U {G_h}(U, \lambda)  $ is singular.
In such a \emph{singular point}, another branch of solutions can bifurcate from the curve of flat solutions and we can use the same continuation method to follow step by step this new curve of steady state solutions.

\medskip

The singular points along the curve of flat solutions can be calculated explicitly both when considering the steady-state FP-LLE~\eqref{eq:mods}  and the discretized problem  $G_h (U,\lambda) =0$ \eqref{eq:1540}, as detailed in Appendix~\ref{app:1125}.
A necessary condition for a bifurcation to non constant steady state solutions,  is given for FP-LLE~\eqref{eq:mods} by the dispersion relation
\begin{equation}\label{eq:3717}
\exists k\in\N^*\quad 
\frac{1}{\beta} \,\Big(
    \alpha - 2\rho_\bullet(\sigma +1)  \pm \sqrt{\rho_\bullet^2-1} \Big) 
=   \frac{k^2}{2}  .
\end{equation}
In that case,  bifurcation points (BP) are such that \hbox{$\rho_\bullet\geq1$.}
If we choose $F$ as continuation parameter for a fixed value of $\alpha$, we deduce from \eqref{eq:3717} that BP on the flat solutions curve are given by $\rho_\bullet$   root of the  polynomial  of  degree 2
\begin{equation*}
(4(\sigma+1)^2-1)\, X^2 - 4(\sigma+1)\,\big(\alpha-\tfrac{1}{2}\beta k^2\big)\,X + \big(\alpha-\tfrac{1}{2}\beta k^2\big)^2+1 .
\end{equation*}
Let $\Delta_\textrm{BP}(k) = \big(\alpha-\tfrac{1}{2}\beta k^2\big)^2 - (4(\sigma+1)^2-1) $.
When $\Delta_\textrm{BP}(k) \geq 0$, the BP are given by
\begin{equation}\label{eq:1426}
\rho_\bullet = \frac{2(\sigma+1)\big(\alpha-\tfrac{1}{2}\beta k^2\big)\pm\sqrt{\Delta_\textrm{BP}(k) }}{  (4(\sigma+1)^2-1) } .
\end{equation}
Thus,  BP along the  flat solutions curve can be obtained by computing the values $\rho_\bullet$ given by \eqref{eq:1426} larger than $1$ when the integer $k$ varies in $\N^*$
and by using the relation
\begin{equation*}
F = \sqrt{\big( 1 + (\alpha - (1+2\sigma)\rho_\bullet)^2\big)\,\rho_\bullet}
\end{equation*}
deduced from \eqref{eq:1518}.
Moreover, pursuing the study on the condition on $k$ for which $\rho_\bullet \geq 1$, one can show that when $\beta < 0$, we have an infinite number of $BP$  whereas on the contrary when  $\beta>  0$, we have an finite number of~$BP$.

 When the continuation parameter is $\alpha$ and $F$ is fixed,  calculations of the BP  are a little bit more tricky since in  equation \eqref{eq:3717}, $\alpha$ is a function of $\rho_\bullet$ given  by (see  \eqref{eq:1518})
 \begin{equation}\label{eq:1509}
\alpha = (1+2\sigma)\,\rho_\bullet \pm \sqrt{\nfrac{F^2}{\rho_\bullet}-1}
\end{equation}
In particular, BP 
occur only for $F^2\geq \rho_\bullet  \geq 1$.
By the change of variable $\rho_\bullet = \frac{F^2}{1+r^2}$, we find that $|r|$ is a root of the  polynomial  of  degree 4
\begin{align}\label{eq:1657}
P_\pm = 4 X^4 &\pm 4\beta k^3\, X^3 + (\beta^2k^4+8)\,X^2\pm (4\beta k^2+8F^2)\,X \nonumber\\
&+ \beta^2k^4+ 4\beta k^2F^2+4
\end{align}
where $P_+$ is related to the $+$ case in relation \eqref{eq:1509} and  $P_-$ is related to the $-$ case.
One can see that $P_+(-X) = P_-(X)$ so that  the roots of $P_-$ are the opposite of the roots of $P_+$. 
Moreover,  studying conditions at which the solution $\rho_\bullet$ to  \eqref{eq:3717}--\eqref{eq:1509} belongs to $[1,F^2]$
makes it possible to show that whatever is the sign of $\beta$, the number of BP along the curve of flat solutions when $\alpha$ is the continuation parameter  is finite.

 \medskip
 
The continuation method outlined here has been implemented to deal with a wide variety of equations in the \textsc{Matlab} Toolbox  \texttt{pde2path} \cite{Uecker:21}
and we have used it to  compute the branches of steady state  solutions to the FP-LLE
that bifurcate from the curve of flat solutions when either $F$ or $\alpha$ 
is used as a  bifurcation parameter. Illustrations are provided in Section~\ref{sec:4}.

%=========================================================
\section{Numerical illustrations}\label{sec:4}
%========================================================

To illustrate various properties of the numerical methods, we now present numerical simulations.
We  focus the presentation on the FP-LLE for the parameter values $\beta=-0.2$ and $F=1.6$ 
in order
to compare  with results obtained in \cite{Mandel:17} for the RS-LLE.
Similar comments to those made with this particular setting could be made with other parameter values. Additional numerical experiments can be undertaken 
 thanks to the open-sources codes  \cite{Balac:22a, Balac:23, Uecker:21}.
%= = = = = = = = = = = = = = = = = = = = = = = = = = =
\subsection{Bifurcation diagram for $\beta=-0.2$ and $F=1.6$}
%= = = = = = = = = = = = = = = = = = = = = = = = = = =
As an example illustrating the continuation method, using \texttt{pde2path} \cite{Uecker:21},  we consider the case of the FP-LLE where $\beta=-0.2$, $F=1.6$ and $\alpha$ is the  bifurcation parameter.
Note that the choice of the laser pump phase detuning  parameter $\alpha$  as a continuation parameter matches well with the experimental use of the FP cavity to generate Kerr frequency combs by tuning  step-by-step this parameter.
The number of BP along the curve of constant solutions is finite and equal to~14.
The BP computed by looking for the roots of the polynomial defined in \eqref{eq:1657} and by using relation~\eqref{eq:1509} are displayed in Table \ref{table:1659} with the corresponding value of $k$ in equation \eqref{eq:1657}  (or equivalently in equation~\eqref{eq:3717}).
These BP and   the branches that bifurcate from the curve of flat solutions at BP are
depicted in Fig.~\ref{fig:1705} where the  $x$-axis shows the bifurcation parameter $\alpha$ and the     $y$-axis shows the $\mathbb{L}^2$  energy norm  of the solution defined
in \eqref{eq:1117}.
The shape of the  curve of flat solutions  represented in black  is typical as shown in Fig.~\ref{fig:1641}.
The bifurcation diagram has been obtained using a Finite Difference discretization of the interval $[-\pi,\pi]$ with $N=2000$ nodes.
In Fig.~\ref{fig:1705},  the crossing of two solution curves does not mean that the two solutions are the same at the intersection point but only that the two solutions have the same   $\mathbb{L}^2$ norm. By choosing  a different  norm  for the $y$-axis, one would obtain a different picture of the bifurcation diagram.

\begin{table}[ht]
\begin{tabular}{|c|c|l||c|c|l|}
\hline
$\ell$ & $k$ & \hspace*{8mm} $\alpha$&$\ell$ &$k$ & \hspace*{8mm} $\alpha$ \\
\hline
1&5& 1.87547953318&
2&4 & 2.53835192735\\
3&6 & 3.97283765189 &
4&3&4.99011120397 \\
5&7&7.14357727273 &
6&2&7.22722777127 \\
7&7&7.69354056103 &
8& 1&7.70074166946 \\
9&6&6.8724727121 &
10&5&5.92544898526 \\
11&4&5.20621519366 &
12&3&4.72595101432 \\
13&2 &4.44909864888 &
14&1 &4.32300480018 \\
\hline
%
%We have found 2 fold points
%FP for alpha = 7.71269184443
%FP for alpha = 3.51502809999
\end{tabular}
\caption{Bifurcation points (BP) on the curve of flat solutions for  $\beta=-0.2$ and $F=1.6$.
The label $\ell$ refers to the position of the BP on the bifurcation diagram depicted on Fig.~\ref{fig:1705} 
and $k$ is the integer in the dispersion relation~\eqref{eq:3717}.}\label{table:1659}
\end{table}

\begin{figure}
\begin{center}
\includegraphics[width=0.9\linewidth]{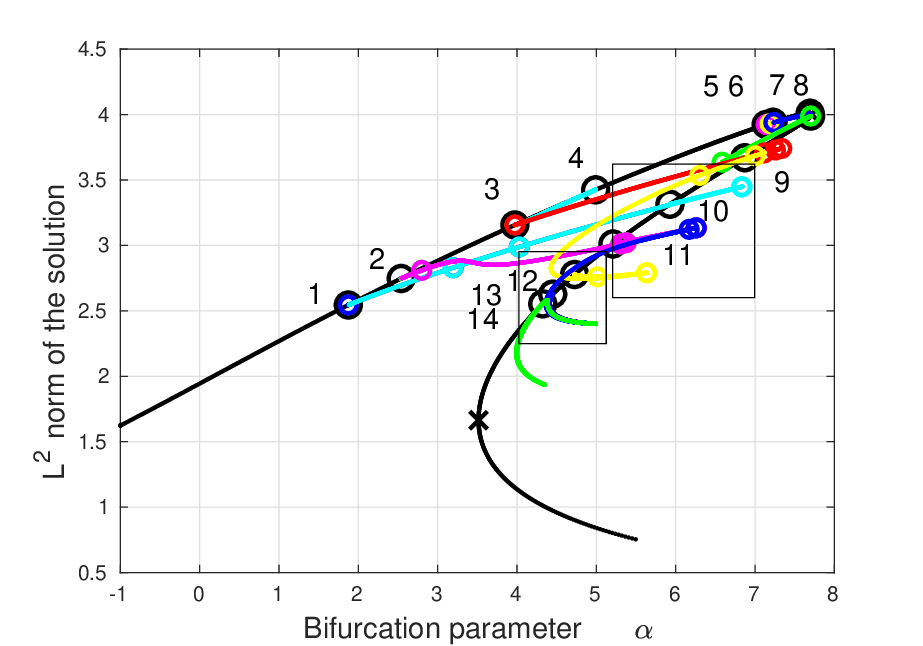}
\includegraphics[width=0.9\linewidth]{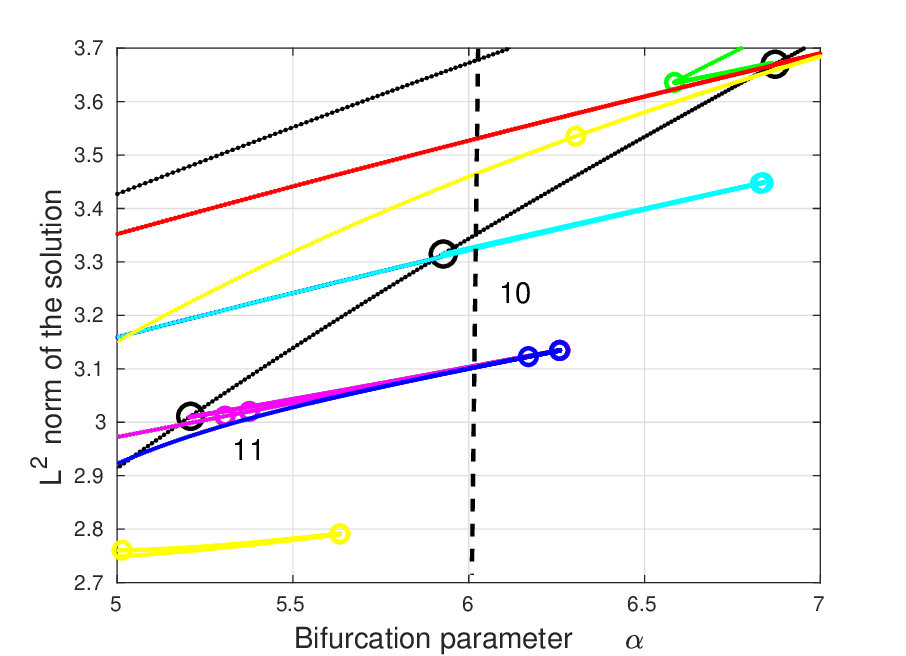}
\includegraphics[width=0.9\linewidth]{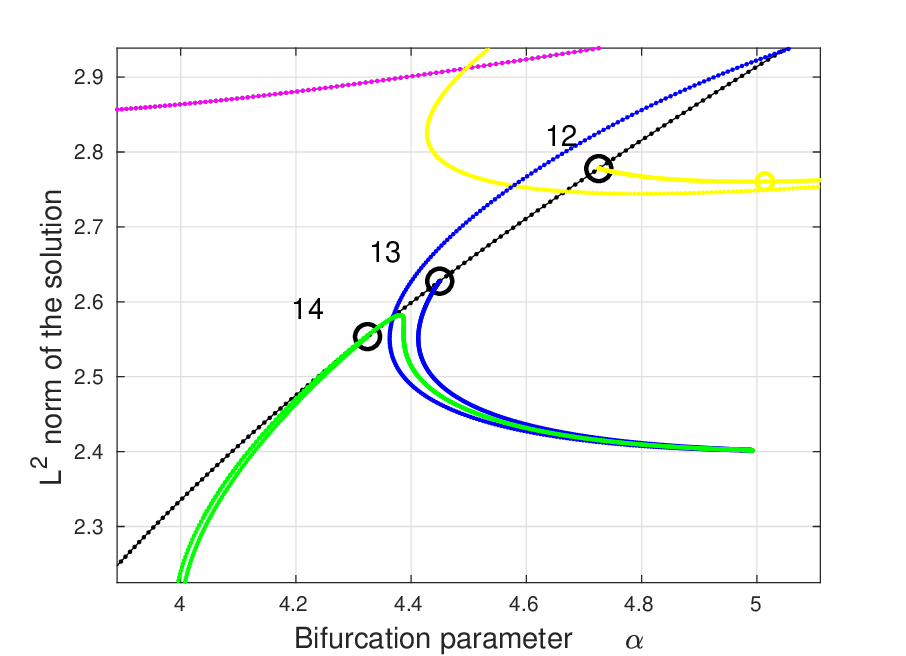}
\caption{Bifurcation diagram. 
Integers refer to the label~$\ell$ of the BP as given in Table  \ref{table:1659}.
Top: BP (black circle) on the curve of flat solutions (black curve)  for  $\beta=-0.2$ and $F=1.6$ and branches of solutions bifurcating from the BP (color curves).
Middle: Zoom on the bifurcation diagram of the larger rectangle area where the bifurcation curves cross the line $\alpha=6$.
 Bottom: Zoom on the bifurcation diagram corresponding to the smaller rectangle area.} \label{fig:1705}
\end{center}
\end{figure}

From the bifurcation diagram in Fig. \ref{fig:1705}, we can understand the reasons why solving the steady state FP-LLE by a Collocation method for a given set of LLE parameters $(\beta, \alpha, F)$ is not easy.
For instance, one can see on  Fig.~\ref{fig:1705}  that for the chosen values of $\beta$ and $F$ ($\beta=-0.2$ and $F=1.6$), we have for the value $\alpha = 6$ at least~$8$ steady state solutions in addition to the three flat solutions (consider the intersections of the vertical line $\alpha=6$ with the curves of solutions and observe that the branches related to BP   with $\ell=$ 10, 11 and 13 have a cusp and cross two times the line $\alpha=6$).
Note that the actual number of steady state solutions is probably  higher since we have only represented here the branches of solutions that bifurcate from the curve of flat solutions while other branches are likely to bifurcate from the  branches issued from the flat solutions curve.
Moreover, one can see from the bifurcation diagram shown at the bottom of   Fig.~\ref{fig:1705}   that the shape of the bifurcation lines can be rather complicated.

The comparison of the bifurcation diagram in Fig.~\ref{fig:1705} for the FP-LLE ($\sigma=1$) to Fig.~5 in \cite{Mandel:17} that represents the bifurcation diagram for the RS-LLE  ($\sigma=0$) for the same set of parameters  ($\beta=-0.2$, $F=1.6$) highlights  similarities and differences between the two kinds of resonators. For instance, one can see that the general shape of the  bifurcation diagrams is similar  and that the number of BP is the same. This can be explained in view of Remark~\ref{rem:1853} p.~\pageref{rem:1853}.
Yet, we can see that with a FP-resonator the range of value for $\alpha$ for which there exist steady state solutions other than flat solutions is wider than this range for a RS-resonator. 
For a RS-resonator, outside the interval $[-0.5, 3.5]$ for $\alpha$, we only have flat solutions  whereas this interval is $[1.5,8]$ for a FP-resonator.
This observation is in accordance with
Proposition~\ref{prop:1443}.

From the steady-state solutions computed by \texttt{pde2path}, we can represent  the corresponding Kerr frequency comb.
For instance, we have depicted in Fig.~\ref{fig:1743}  five steady state solutions and Kerr frequency combs  for the FP-LLE parameters $\beta=-0.2$, $F=1.6$ and $\alpha=6$.
Namely, we have depicted the real, imaginary parts and modulus of the solution $\psi$ as well as the corresponding frequency comb obtained as the sequence
$10\,\log_{10}(| c_n|^2)$ where $(c_n)_{n\in\Z}$ are the Fourier coefficients of $\psi$.

We want to point out that the continuation method provides a thorough knowledge of  steady state solutions to  the FP-LLE  but this requires a long computation time. For instance, obtaining the  bifurcation diagram in Fig.~\ref{fig:1705}   required several hours of computation on a computer work-station.
Moreover, depending on the parameter values it can be difficult to interpret the bifurcation diagram.

\begin{figure}[h]
\begin{center}
\includegraphics[width=0.49\linewidth]{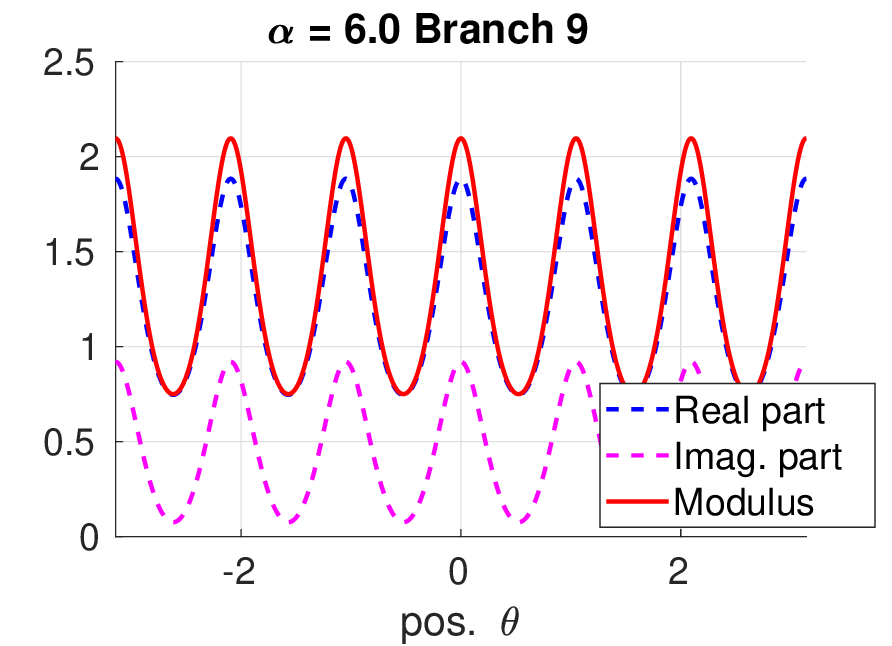}
\includegraphics[width=0.49\linewidth]{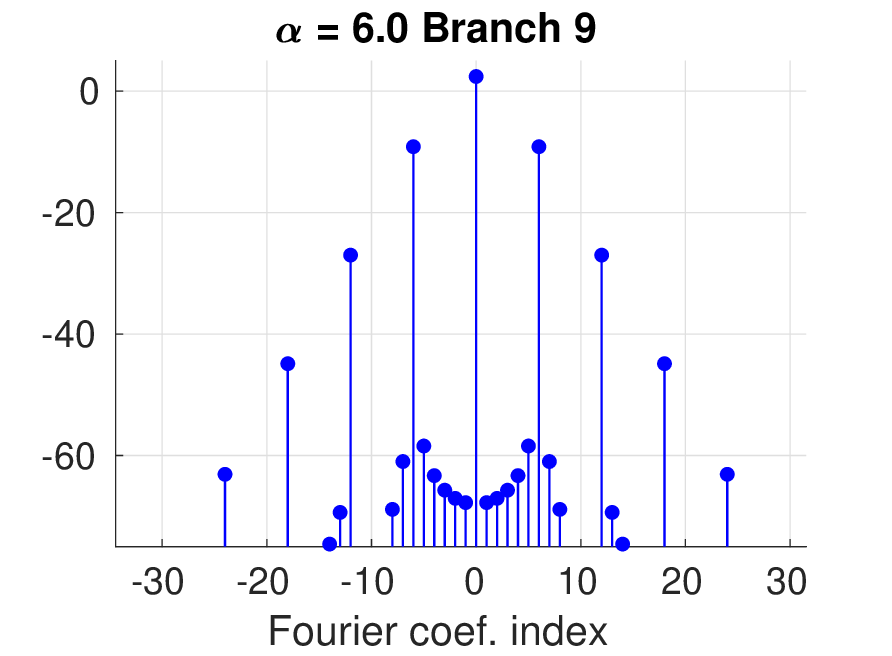}
\includegraphics[width=0.49\linewidth]{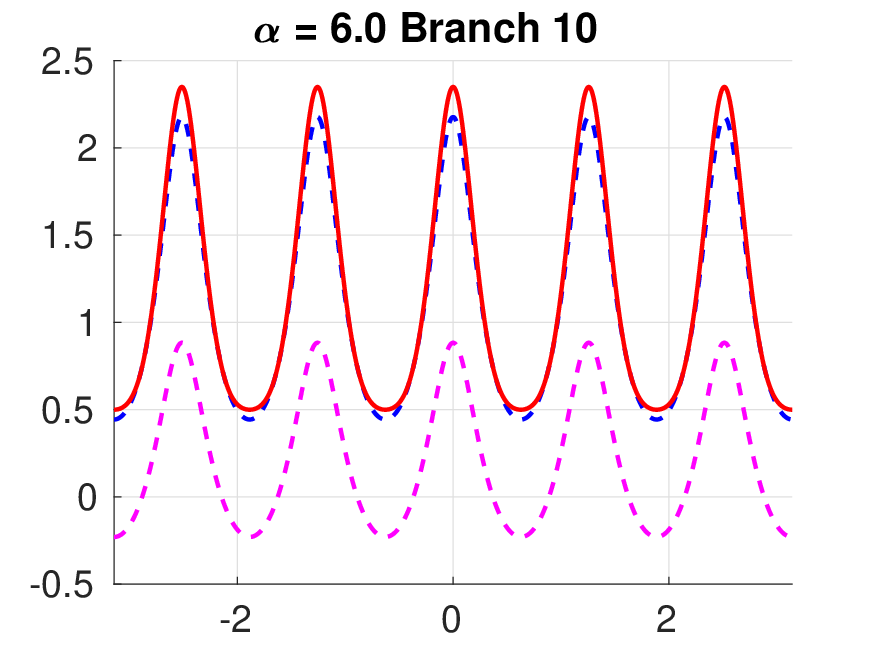}
\includegraphics[width=0.49\linewidth]{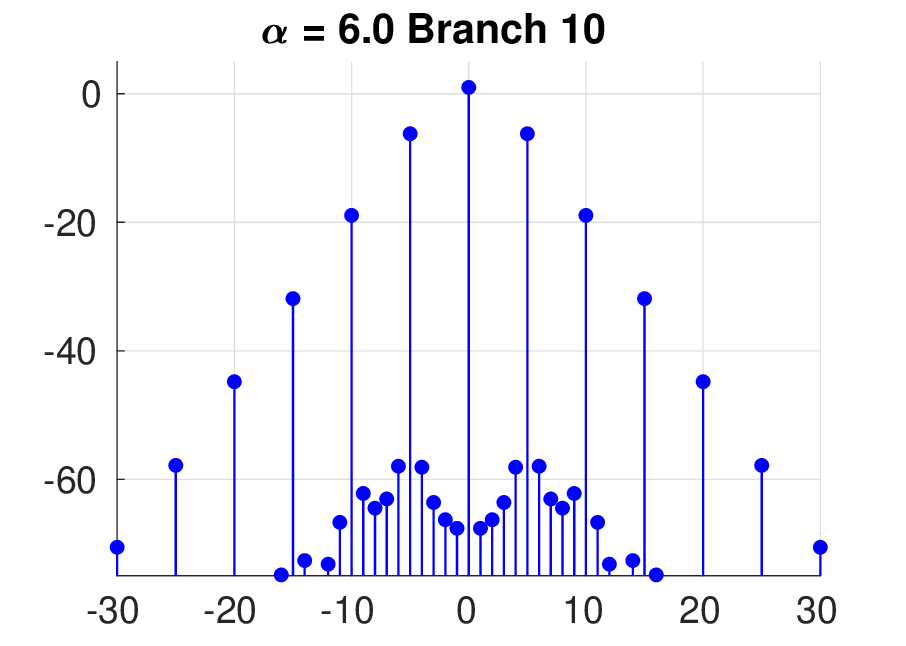}
\includegraphics[width=0.49\linewidth]{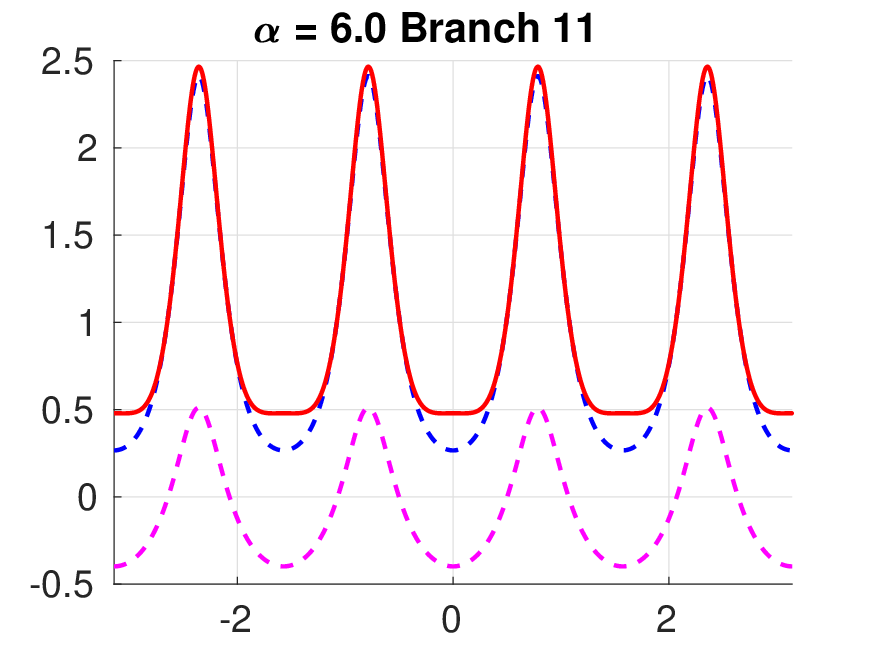}
\includegraphics[width=0.49\linewidth]{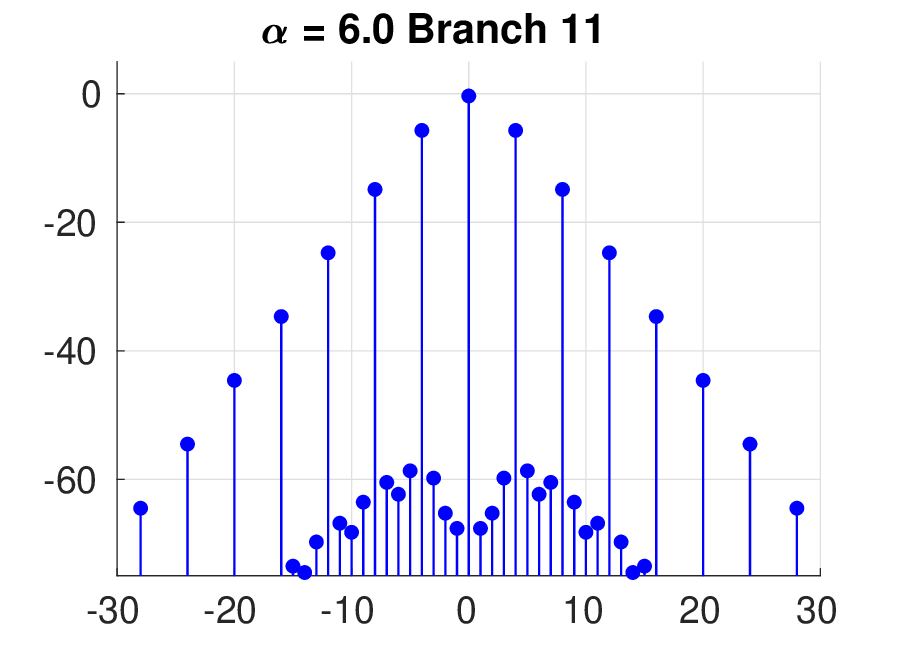}
\includegraphics[width=0.49\linewidth]{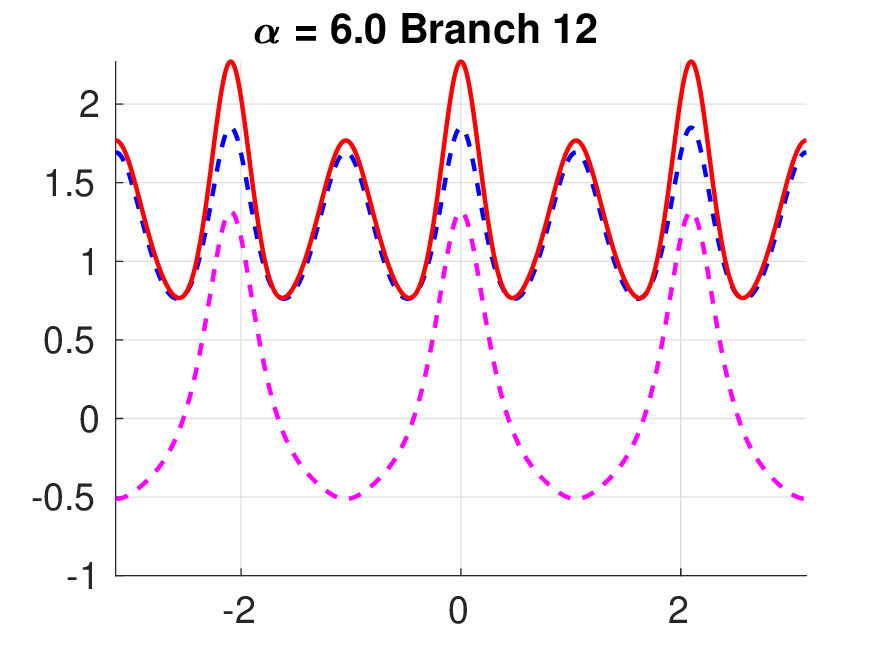}
\includegraphics[width=0.49\linewidth]{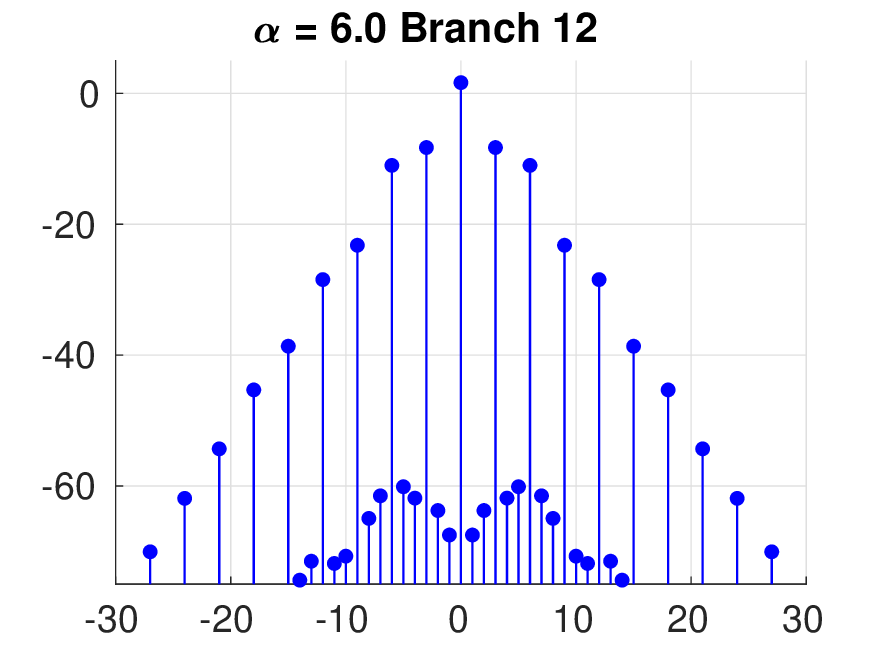}
\includegraphics[width=0.49\linewidth]{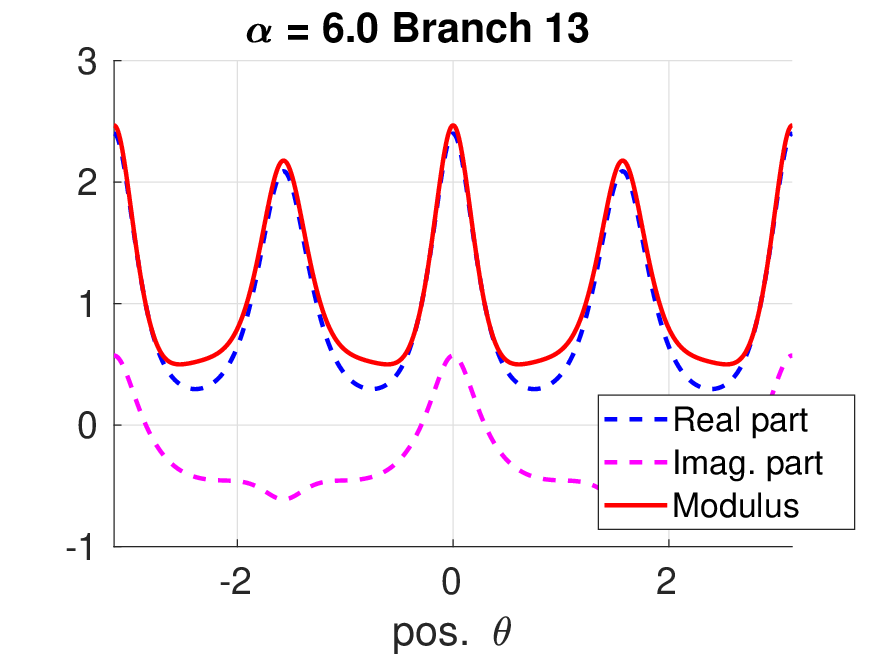}
\includegraphics[width=0.49\linewidth]{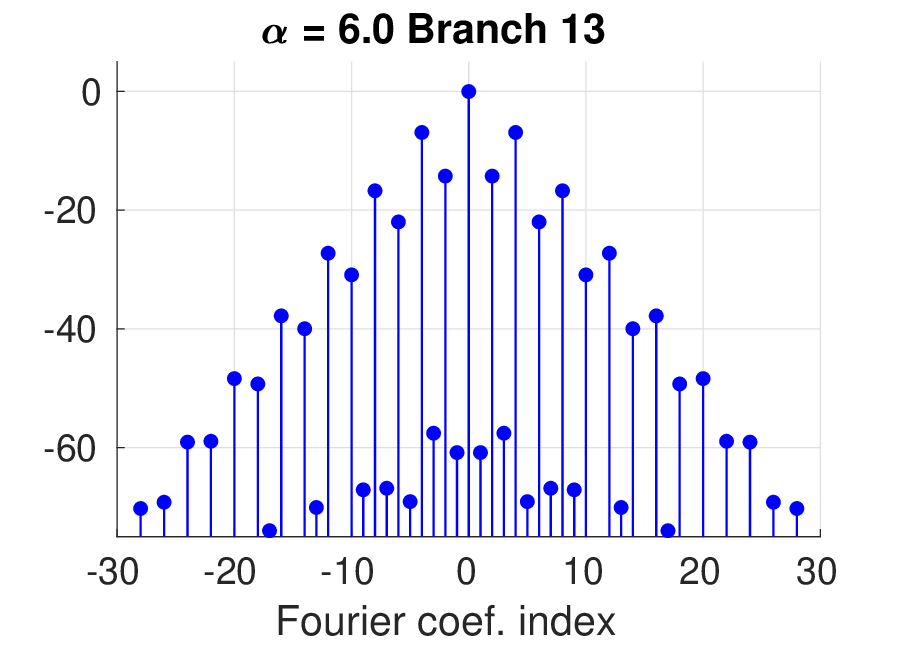}
\caption{Solutions to the FP-LLE for $\alpha=6, \beta = -0.2$ and $F=1.6$  (left) and corresponding Kerr frequency combs (right) obtained on the bifurcation branches starting from BP 9, 10, 11, 12 and 13 (from top to bottom).
On the left, the real part of the solution is drawn in blue dashed line, its imaginary part in magenta dashed line  and it modulus in red solid line.
The frequency comb on the right is a representation of the sequence
$10\,\log_{10}(| c_n|^2)$ where $(c_n)_{n\in\Z}$ are the Fourier coefficients of the solution.
}\label{fig:1743} 
\end{center}
\end{figure}

%= = = = = = = = = = = = = = = = = = = = = = = = = = =
\subsection{Solutions to the dynamic LLE for $\beta=-0.2$, $F=1.6$ and $\alpha=6$}
%= = = = = = = = = = = = = = = = = = = = = = = = = = =
Let us now consider  the Split-Step method presented in Section \ref{sec:S3F4LLE} and implemented in \textsc{Matlab} Toolbox  \texttt{S3F4LLE}~\cite{Balac:22a} to 
solve the time-dynamic FP-LLE~\eqref{eq:mod}.
From the physics  point of view,  a precise knowledge of the initial state of the  FP cavity (described by the function $\psi_0$  in the initial condition  \eqref{eq:IC}) is merely impossible and  for the purpose of numerical simulation, we   have to make the choice of  manufactured initial conditions.   
Let us consider for instance an initial condition corresponding to a finite  lattice  sum of Gaussian functions
\begin{equation}\label{eq:1510}
g_{m,P} (\theta) = \sum_{p=-P}^P \sum_{\ell=1}^m \e^{- \frac{1}{|\beta|}\big( \theta+2p\pi + \frac{2\ell-m-1}{m}\pi\big)^2} .
\end{equation}
Here the main parameter is the integer $m$ aimed at exciting frequency combs of width $k=m$. The integer parameter $P$ allows, function of the value of $\beta$, a precise approximation of exact periodicity at $\theta=\pm\pi$.

We consider some of the solutions corresponding to the bifurcation diagram depicted in Fig.~\ref{fig:1705}.
We have depicted in Fig.~\ref{fig:1104}  the time evolution of the  solution to the FP-LLE~\eqref{eq:mod}  for $\alpha=6, \beta = -0.2$ and $F=1.6$   when the initial condition is $\psi_0 = g_{4,2} + \psi_{\bullet,2}$ where  $ \psi_{\bullet,2}$ denotes the flat solution with intermediate value among the 3  flat solutions of the FP-LLE for this set of parameters.
Simulation is carried out until normalized time $T=100$ with $10^5$  sampling points
and $2^8$ space  sampling points for the FFT computations.
One can see that after strong variations of the solution over a short duration (less than 10 units of time), the solution converges toward a steady state solution that has the features of the steady state solution observed on branch 11, see Fig.~\ref{fig:1743}.
The spacing of the comb teeth is 4 units in accordance with the value of $k$ given in Table~\ref{table:1659}.
Compared to the comb depicted in Fig.~\ref{fig:1743} for branch 11, the  pattern at the bottom of the comb  in Fig.~\ref{fig:1104}  has disappeared which shows that the bottom pattern is a numerical artifact.

\begin{figure}[h]
\begin{center}
\includegraphics[width=0.49\linewidth]{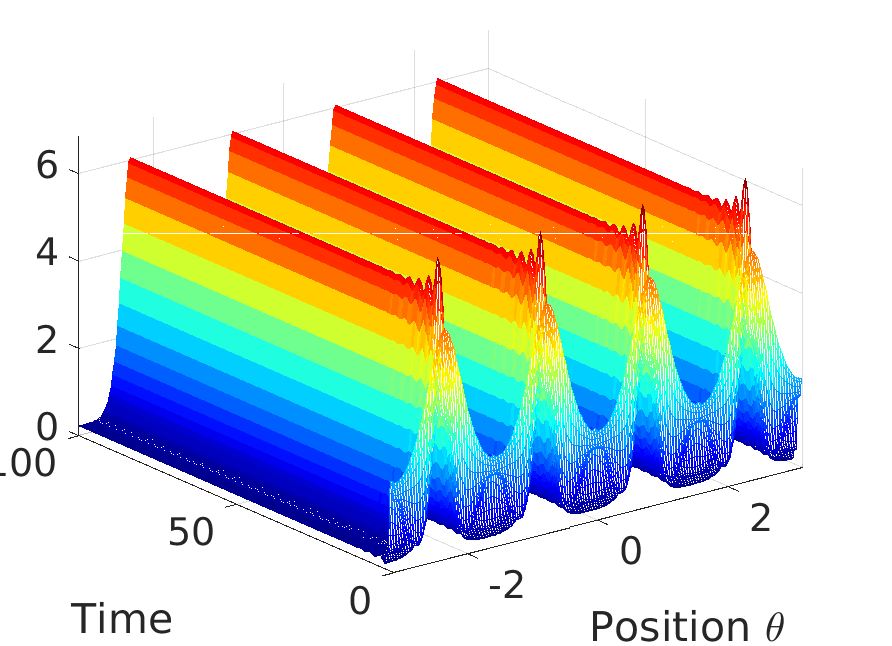}
\includegraphics[width=0.49\linewidth]{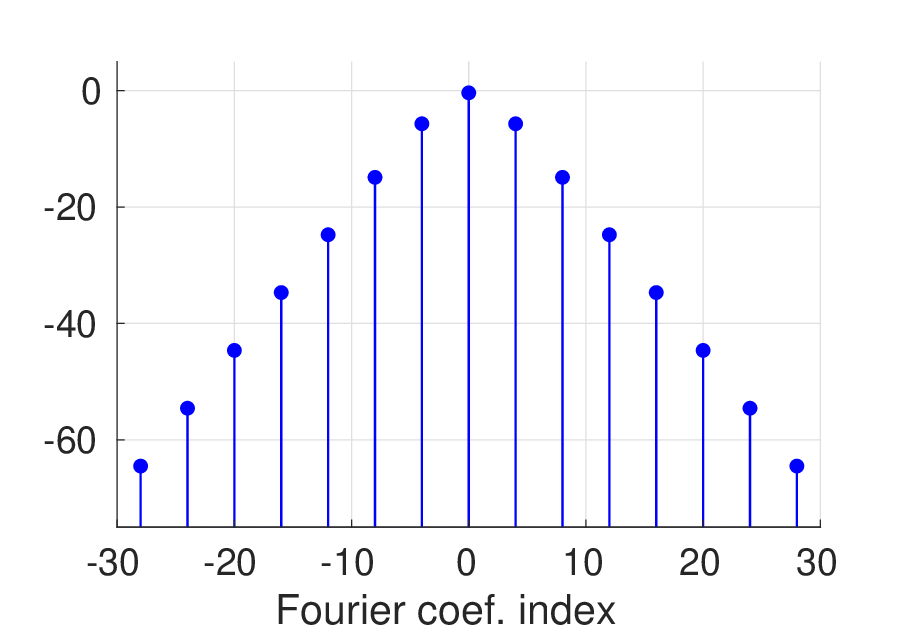}
\caption{Left:
Time evolution of the square modulus   $|\psi|^2$ of the  solution $\psi$ to the time-dynamic  FP-LLE for $\alpha=6, \beta = -0.2$ and $F=1.6$  
with  initial condition corresponding to a finite  lattice  sum of Gaussian functions  \eqref{eq:1510} with $P=2$ and $m=4$.  Right: Corresponding Kerr frequency comb.}\label{fig:1104} 
\end{center}
\end{figure}

\medskip

Note that not all the steady state solutions computed by the continuation method are stable.
The question of the stability of steady state solutions to the RS-LLE ($\sigma=0$) is investigated from a mathematical point of view in \cite{Miyaji:10}.
For illustrative purposes, we have depicted in Fig.~\ref{fig:2104} the  solution to the time dynamic FP-LLE~\eqref{eq:mod} at time $100$  with the same set of parameters as those of  Fig.~\ref{fig:1104} excepted that $m=5$ in the  initial condition given by the finite  lattice  sum of Gaussian functions.
One can see that this solution matches well with the steady state solution obtained by the continuation method for $\alpha=6$ on branch $10$,  see Fig.~\ref{fig:1743}.
However, when the evolution of the solution is observed on a duration of 250 units of time,  one can see that it seems to oscillate (or to hesitate) between two shapes (that are very close to each other) but finally around time $210$ the solution tends to the smaller of the 3  flat solutions, see Fig.~\ref{fig:2105}.
 The same observation can be made when using as initial condition the finite  lattice  sum of Gaussian functions with $m=3$ or $m=6$ where the solution  firstly  look very close to the steady state solution observed on branch 9 before to suddenly change its behavior and to converge to  the smaller   flat solution.

\begin{figure}[h]
\begin{center}
\includegraphics[width=0.49\linewidth]{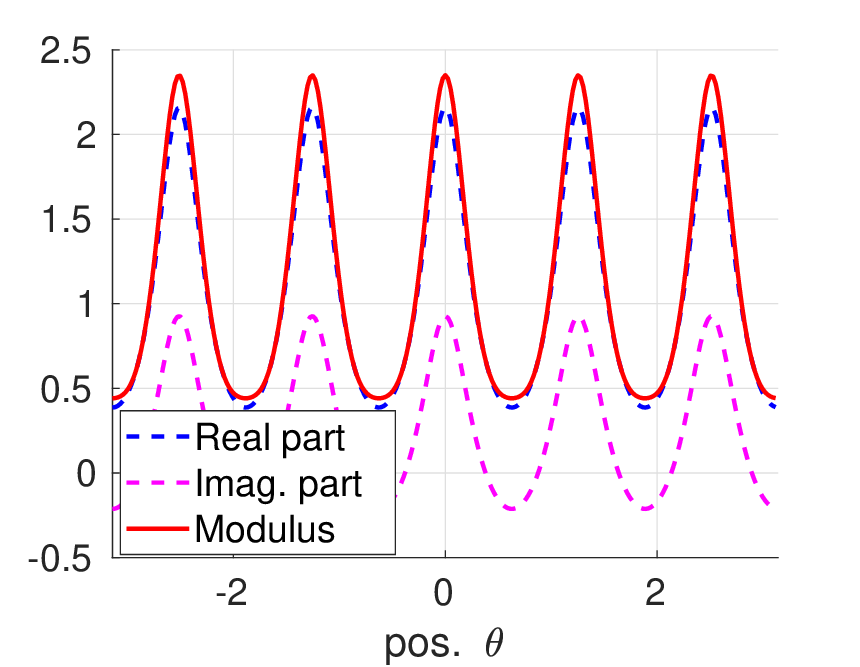}
\includegraphics[width=0.49\linewidth]{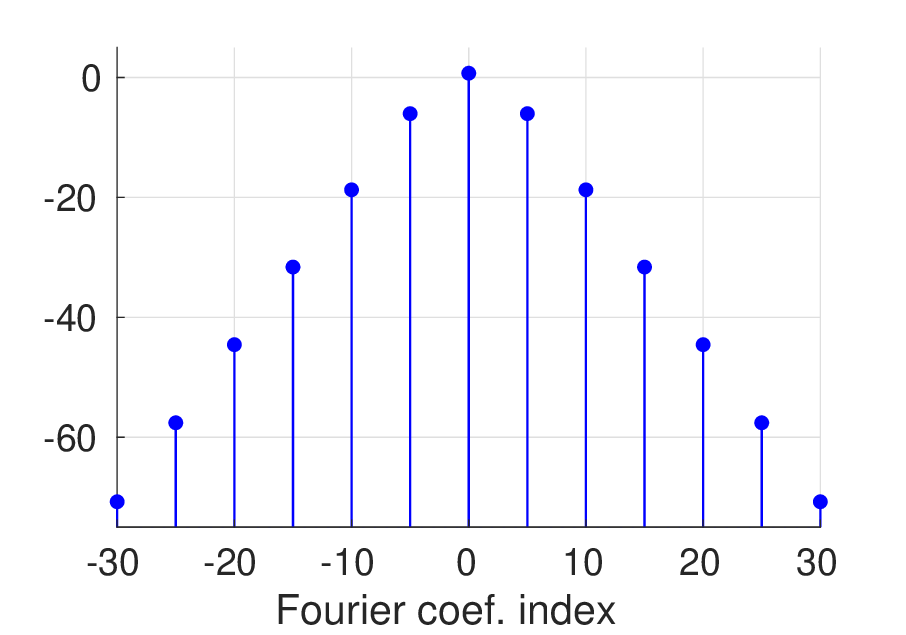}
\caption{Left: Solution to the time-dynamic  FP-LLE for $\alpha=6, \beta = -0.2$ and $F=1.6$  
with  initial condition corresponding to a finite  lattice  sum of Gaussian functions  with $P=2$ and $m=5$ at time $100$.  Right: Corresponding Kerr frequency comb.}\label{fig:2104} 
\end{center}
\end{figure}

\begin{figure}[h]
\begin{center}
\includegraphics[width=0.49\linewidth]{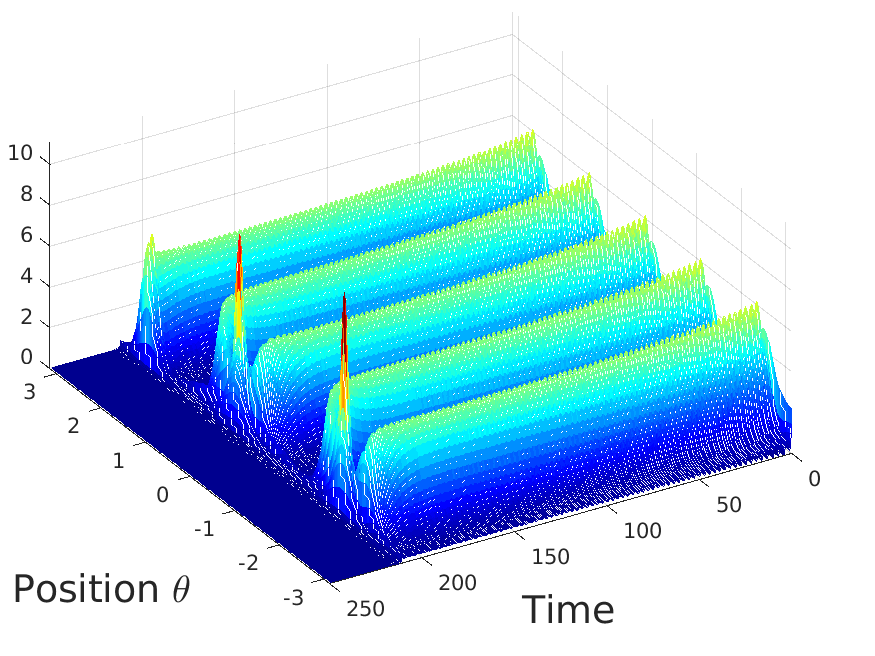}
\includegraphics[width=0.49\linewidth]{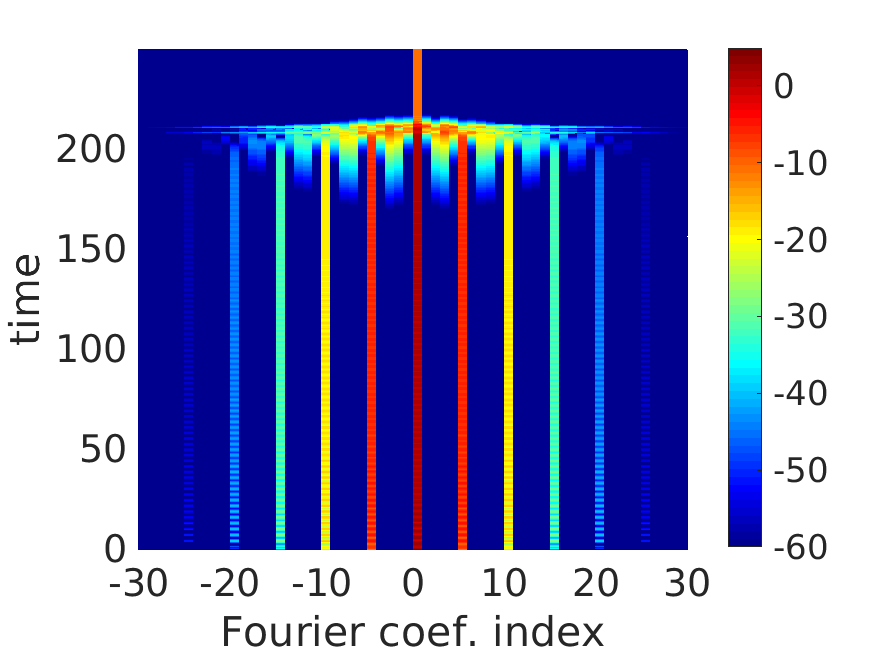}
\caption{Left:  Time evolution  of the square modulus of the solution to the time-dynamic  FP-LLE for the parameters of Fig.~\ref{fig:2104}.   Right:  Variation of the solution spectrum with time.}\label{fig:2105} 
\end{center}
\end{figure}

\medskip

As a last illustration, let us consider the set of parameters $\beta=-0.2$, $F=1.6$ and $\alpha=4$
with initial condition $\psi_0 = g_{6,2} + \psi_{\bullet,2}$ where  $ \psi_{\bullet,2}$ denotes the flat solution with intermediate value among the 3  flat solutions of the FP-LLE for this set of parameters
and $g_{2,2}$  the finite  lattice  sum of Gaussian functions.
We have depicted in Fig.~\ref{fig:0955}  the  evolution of the square modulus of the  solution  with time and we can see  four distinct behaviors for the solution from the initial condition to the steady state that appears from time $150$. 
By using the bifurcation diagram given in  Fig.~\ref{fig:1705} and the various solution shapes along the bifurcation branches  depicted in Fig.~\ref{fig:1743}, one can infer that the solution has converged to the steady state solution of branch 13. 
It is also interesting to run a simulation with  initial condition $\psi_0 = g_{2,2} + \psi_{\bullet,3}$ where  $ \psi_{\bullet,3}$ denotes the flat solution with higher value.
The evolution of the solution with time also exhibits  four distinct stages but this time the solution converges to the steady state solution of branch 11 as illustrated in Fig.~\ref{fig:0956}.
One can note by comparing Fig.~\ref{fig:0956} and Fig.~\ref{fig:0955}  that when the steady state regime is reached, the two combs are identical and that the two solutions are translated from each other. This observation is in accordance with the remark made in Section~\ref{sec:21}  regarding the translation invariance properties of steady state solution.

\begin{figure}[h]
\begin{center}
\includegraphics[width=0.49\linewidth]{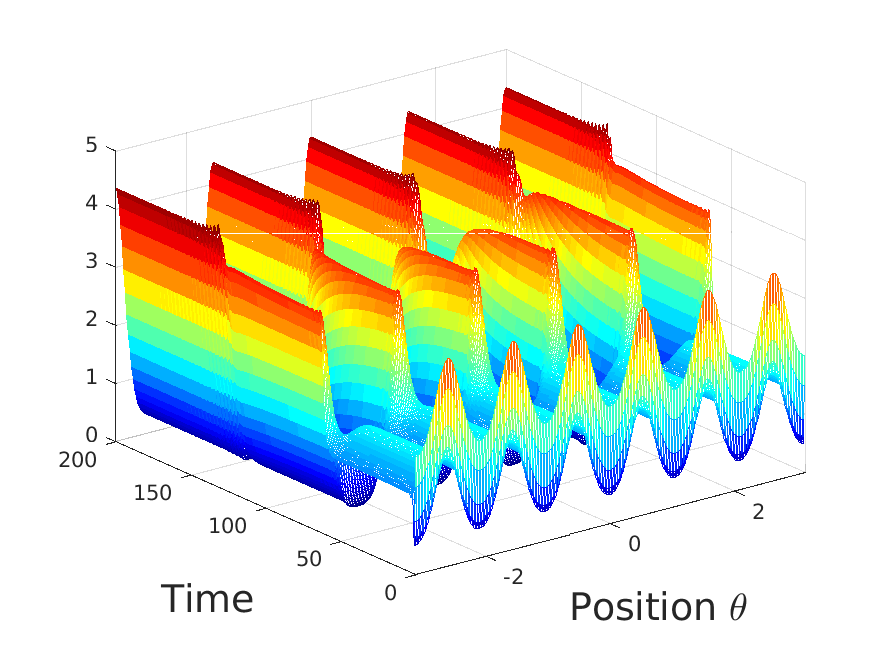}
\includegraphics[width=0.49\linewidth]{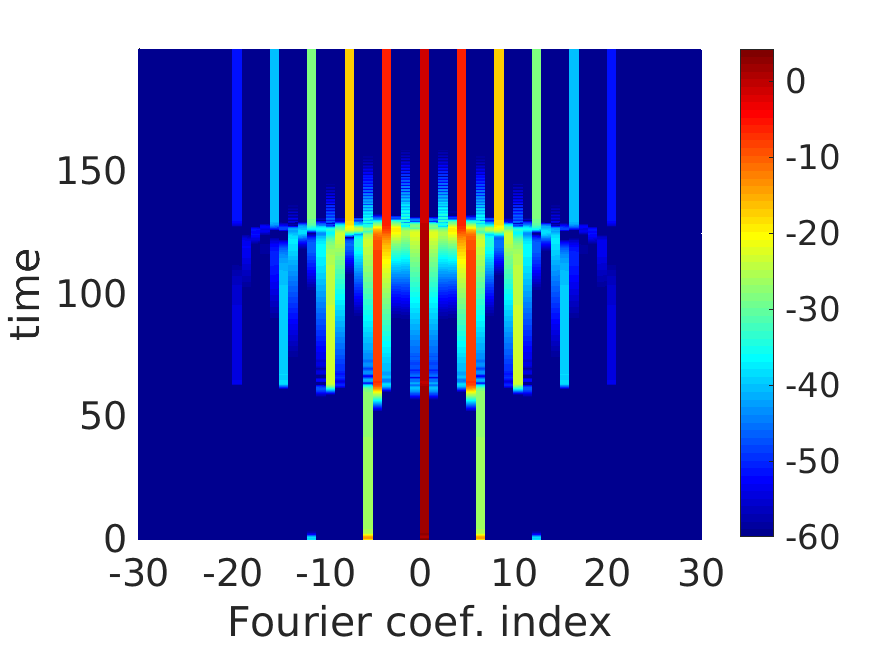}
\includegraphics[width=0.49\linewidth]{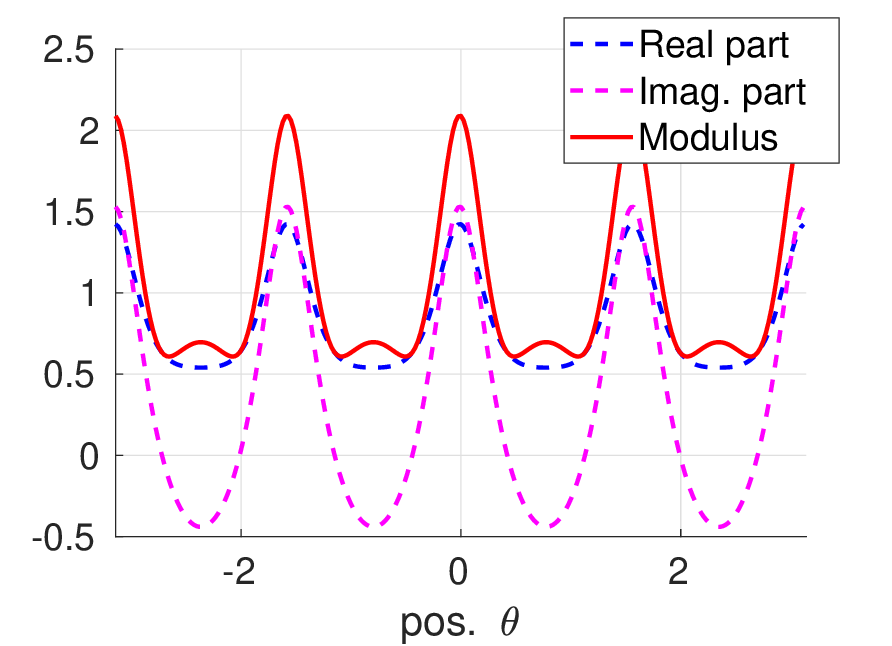}
\includegraphics[width=0.49\linewidth]{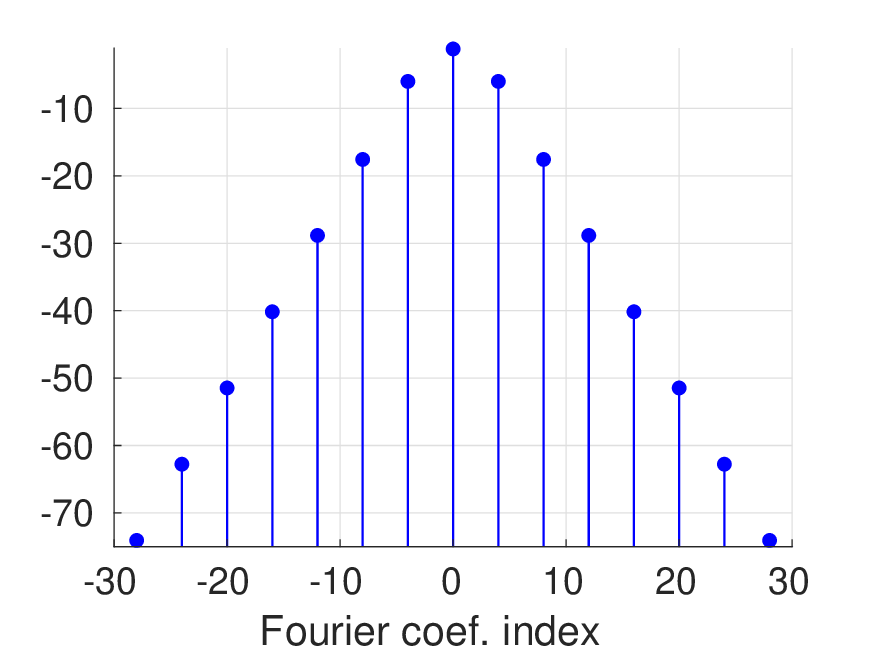}
\caption{Top-left:  Time evolution of the square modulus of the  solution to the time-dynamic  FP-LLE for $\beta=-0.2$, $F=1.6$ and $\alpha=4$
with initial condition $\psi_0 = g_{6,2} + \psi_{\bullet,2}$.   Top-right:  Variation of the solution spectrum with time.
Bottom-left: Solution at final time $T=200$. Bottom-right: Kerr frequency comb of the steady state solution.
}\label{fig:0955} 
\end{center}
\end{figure}

\begin{figure}[h]
\begin{center}
\includegraphics[width=0.49\linewidth]{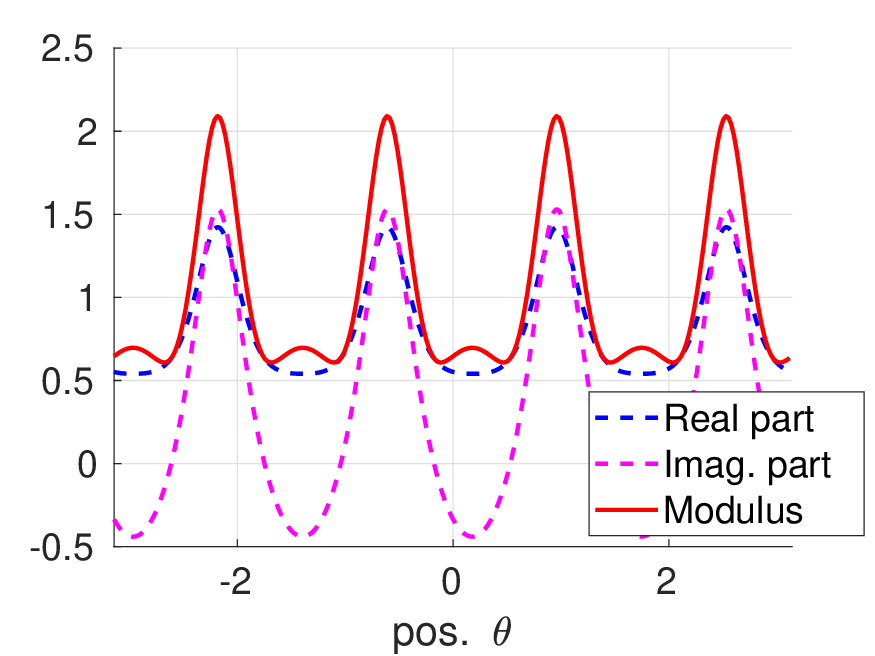}
\includegraphics[width=0.49\linewidth]{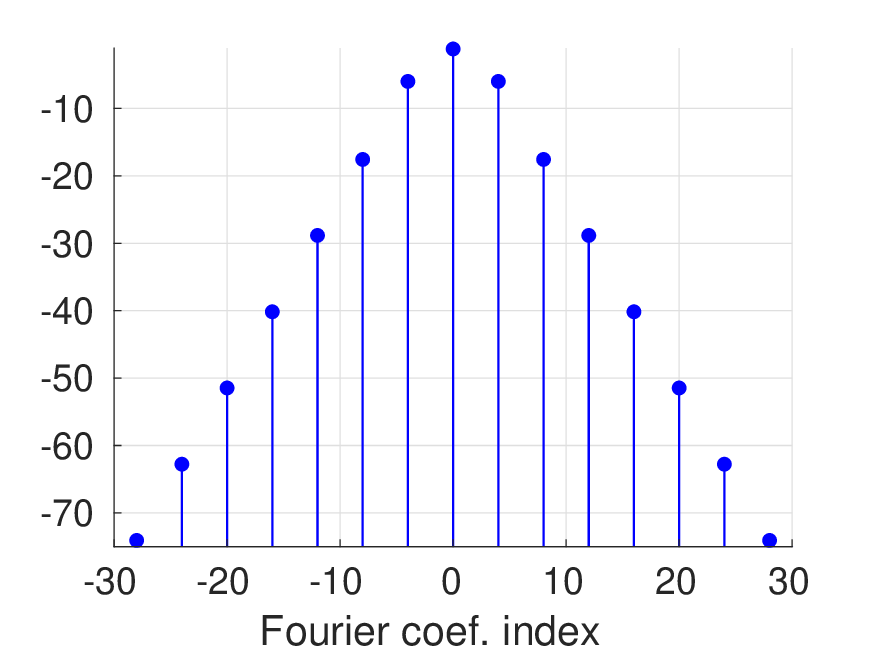}
\caption{Left:  Solution at time $T=200$ to the  FP-LLE for  $\psi_0 = g_{6,2} + \psi_{\bullet,3}$ with other parameters as in Fig.~\ref{fig:0955}.   Right:  Kerr frequency comb of the steady state solution.
}\label{fig:0956} 
\end{center}
\end{figure}

\medskip

In view of the results presented above, the question of whether a steady state has been obtained or not is major but not easy to answer from the graphical representation of the  solution.
In order to have a more effective indicator, we propose to compute at each time-step $t_n$ the quantity 
\begin{equation}\label{eq:1502}
 \mathcal{V}(\psi,t_n) = \max_{\theta\in[-\pi,\pi]} \left| \frac{\psi(\theta,t_n) - \psi(\theta,t_{n-1} ) } { t_n-t_{n-1}} \right| .
\end{equation}
This numerical indicator   shows  how the computed solution $\psi$  varies  from a time-step  to the other and therefore, values  of $ \mathcal{V}(\psi,t_n)$ closed to the machine epsilon indicate that a steady state regime has been reached. 
We have depicted in Fig.~\ref{fig:1059} the variations with time of the decimal logarithm of   $ \mathcal{V}(\psi,t_n) $  for the simulation reported in Fig.~\ref{fig:0955}.
One can easily identify the various states the solution goes through in the interval $[0,150]$ before converging  to its final steady state that is fully 
attained  around time $400$.

\begin{figure}[h]
\begin{center}
\includegraphics[width=0.75\linewidth]{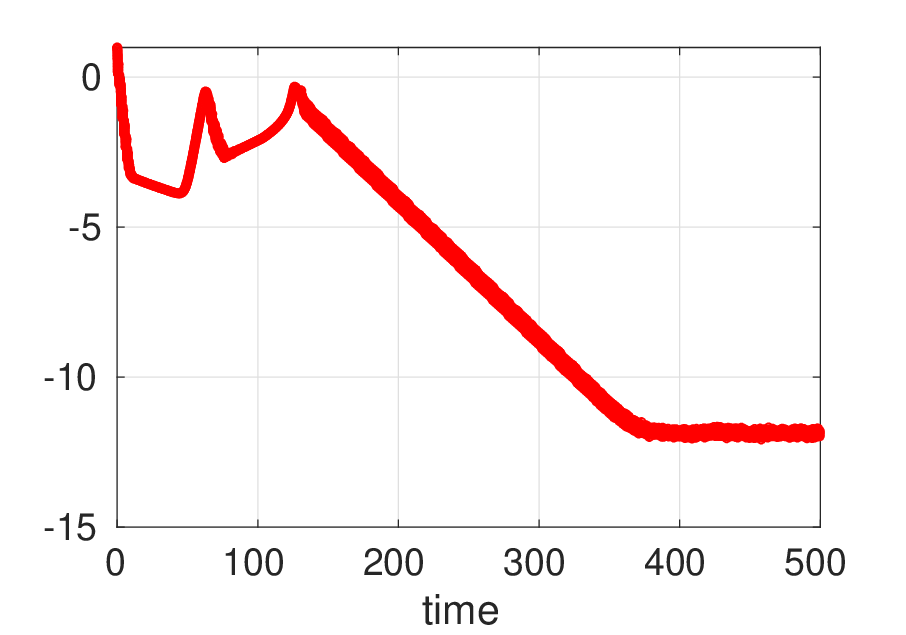}
\caption{Variations with time of the decimal logarithm of the indicator $ \mathcal{V}(\psi,t_n) $  for the simulation reported in Fig.~\ref{fig:0955}.
}\label{fig:1059} 
\end{center}
\end{figure}

To conclude, in connection with Kerr comb generation, we can stress once again the leading  role played by  the  initial data $\psi_0$  on the type of steady state solution attained (even if no  simple rule of thumb can be formulated)
and the need of a  steady state indicator such as the one defined in \eqref{eq:1502} to have the assurance that a  steady state has been reached when the simulation is ended.

%= = = = = = = = = = = = = = = = = = = = = = = = = = =
\subsection{Stationary solutions computed by the Collocation method}
%= = = = = = = = = = = = = = = = = = = = = = = = = = =

In \cite[Section V]{Cole:18}, the authors provide an analytical approximate  expression for Soliton solution to the FP-LLE. (Note that the use of this expression requires the solving of a non-linear equation and therefore necessitates a numerical solver anyway.)
When they exist, Solitons solutions to the  FP-LLE can be computed  more accurately by the Collocation method described in Section \ref{sec:1640}
and implemented in our   \textsc{Matlab} toolbox \texttt{COLLE} \cite{Balac:23}.
For instance we have depicted in Fig.~\ref{fig:1501}  the Soliton solution  to the  FP-LLE for the values  $\beta=-0.02$, $F=\sqrt{3}$ and $\alpha=4.37$ used in  \cite{Cole:18}, obtained by the Collocation method for an initial guess $\psi_{\bullet,1}+g_{1,4}$  where  $ \psi_{\bullet,1}$ denotes the flat solution with lower value and $g_{1,4}$ the finite  lattice  sum of Gaussian functions~\eqref{eq:1510}
and a number of discretization nodes $N=10^3$.
The CPU time required for this computation is 9.1~s.
The relative error between the approximate analytical expression of the Soliton proposed in \cite{Cole:18}
and the solution computed by the Collocation method  measured in the energy norm is $5.9\%$.

\begin{figure}[h]
\begin{center}
\includegraphics[width=0.49\linewidth]{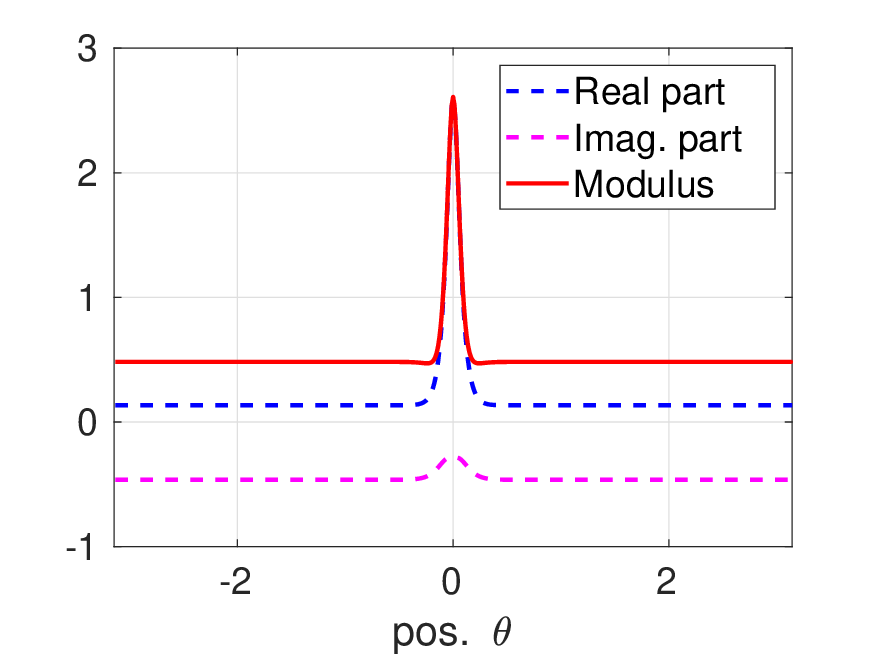}
\includegraphics[width=0.49\linewidth]{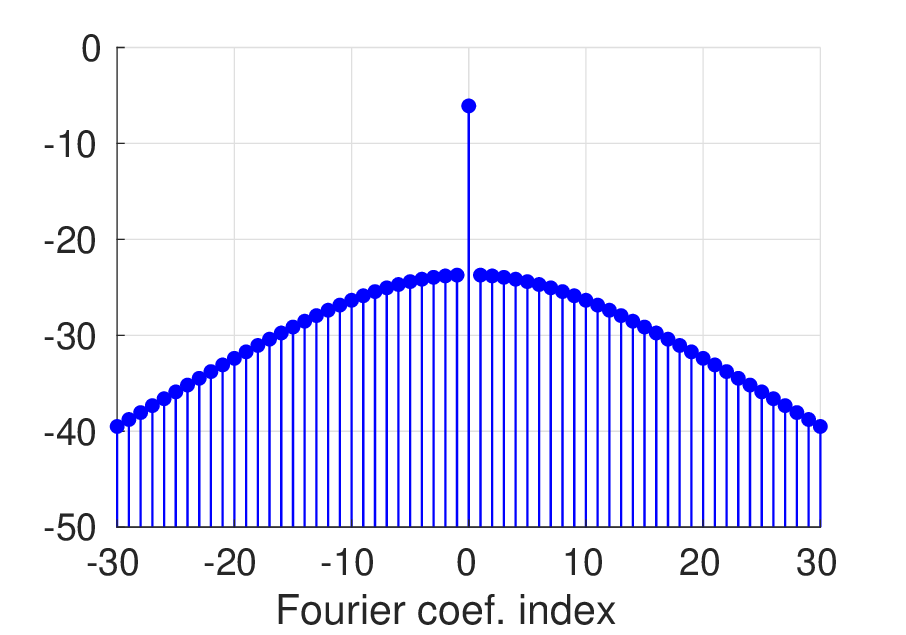}
\caption{Left:  Soliton solution  to the  FP-LLE for $\beta=-0.02$, $F=\sqrt{3}$ and $\alpha=4.37$.
Right:  Kerr frequency comb of the Soliton.
}\label{fig:1501} 
\end{center}
\end{figure}

\medskip

It is also possible by using the Collocation method to compute solutions to the FP-LLE referred to as  Turing patterns in \cite[Section VI]{Cole:18}  and observed for $\beta=-0.02$, $F=\sqrt{6}$ and $\alpha=2.5$.
 We have depicted in Fig.~\ref{fig:1606}  the  solution computed  by the Collocation method for an initial guess $\psi_0: \theta\mapsto \psi_{\bullet,1}+ \cos(m \theta)$ for $m=16$.
 
\begin{figure}[h]
\begin{center}
\includegraphics[width=0.49\linewidth]{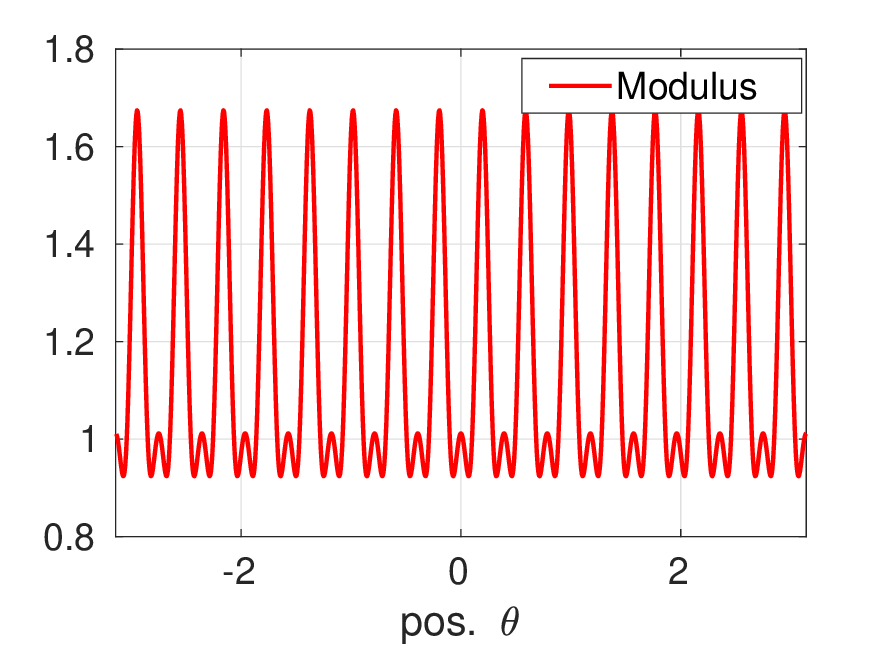}
\includegraphics[width=0.49\linewidth]{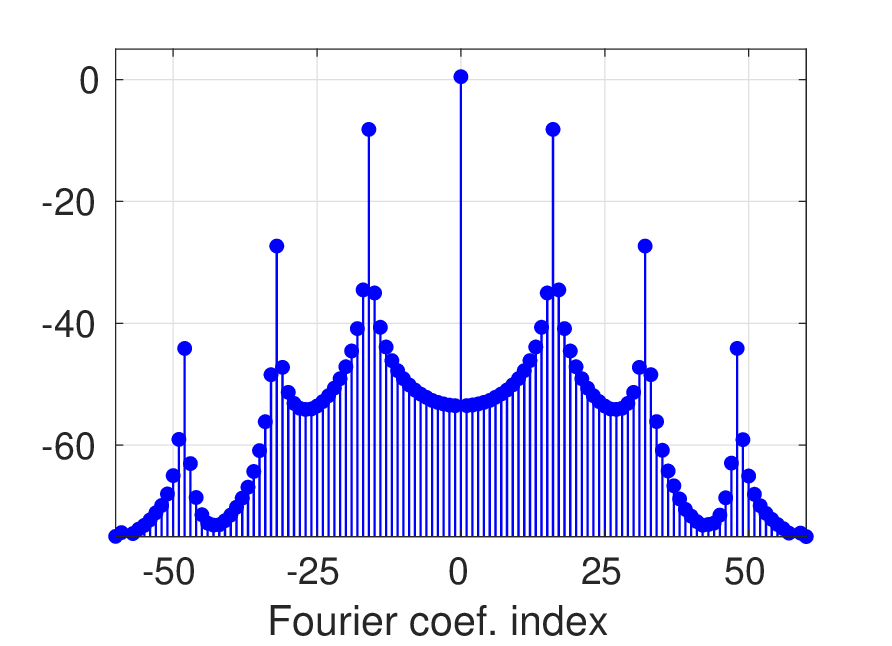}
\caption{Left:  Turing pattern  solution  with $16$ rolls  to the  FP-LLE for $\beta=-0.02$, $F=\sqrt{6}$ and $\alpha=2.5$.
Right:  Kerr frequency comb of the computed  Turing pattern. The main teeth of the comb are located every 16 units.  
}\label{fig:1606} 
\end{center}
\end{figure}

%====================
\section{Conclusion}
%====================

We have presented three numerical approaches to study by numerical simulation Kerr frequency combs in optical resonators.
We have considered two variants of the  Lugiato-Lefever equation (LLE)  that were introduced in the literature to  model respectively light-wave propagation in ring-shape  (RS) resonators   and Fabry-Perot (FP) resonators. We have focused our attention on the  LLE for FP resonators because this topic has been introduced recently and  is presently  extensively investigated in optics. The particularity of  the FP-LLE  compared to the RS-LLE is  an additional integral term that makes the LLE non-linearity spatially non-local and consequently  makes its numerical resolution more difficult. 

 The three numerical approaches we have described in the paper include a Split-Step method to solve the time-dynamic FP-LLE and a Collocation method to solve the steady state FP-LLE.
These methods are very efficient but the computed solution highly depends on the initial condition when solving   the   time-dynamic FP-LLE  or on the initial guess solution when  solving  the steady state LLE by the Collocation method and these initial data are not easy to know or to guess.
 We have pointed out in the paper the specific role played by  flat solutions, \textsl{i.e.} spatially and temporally constant solutions,  and we found that they are the solutions the most often observed during numerical simulations if the initial data has not the specific features to induce convergence to a non-flat steady state solution.
 In order to overcome this issue, we have presented a continuation method with respect to the parameters $\alpha$ or $F$ of the LLE that exhibits bifurcation branches from flat solutions.
 We obtain a comprehensive picture of the various solutions to the steady state LLE 
 and this constitutes an  important help in the interpretation of the results provided by   the
 Split-Step and the Collocation methods.
 
 The numerical methods presented in this paper to solve the FP-LLE have been implemented under \textsc{Matlab}  and specific toolboxes have been developed and are distributed as free software, see \cite{Balac:22a, Balac:23}

%==========================================================
\section*{Acknowledgments}
%==========================================================

This work has been undertaken in the framework of the ANR project ROLLMOPS (optical resonator with ultra-high quality factor for high spectral purity microwave signals generation) funded by the French Agence Nationale de la Recherche (2021-2024).

\appendix
%====================
\section{Energy estimate for the solution to the time dynamic FP-LLE}\label{app:1124}
%====================
Let us consider a solution $\psi$ to  the time dynamic FP-LLE \eqref{eq:LLE}.
By multiplying \eqref{eq:LLE} by the solution conjugate~$\overline{\psi}$  and integrating over $[-\pi,\pi]$, we obtain
\begin{align}\label{eq:1115}
&\int_{-\pi}^{\pi} \!\!
 \partial_t \psi(\theta,t)   \, \overline{\psi}(\theta,t) \dd\theta 
  =  
 -\ic \frac{\beta}{2} \,   \int_{-\pi}^{\pi} \!\! \partial^2_{\theta\theta} \psi (\theta,t) 
 \, \overline{\psi}(\theta,t) \dd\theta 
 \nonumber\\
&\hspace*{5mm} 
 - (1+\ic \alpha)\, \|\psi(t)\|^2 + \ic \| \psi^2(t)\|^2  + \ic \frac{\sigma}{\pi}\,  \|\psi(t)\|^4 
  \nonumber\\
&\hspace*{5mm} 
+   F    \int_{-\pi}^{\pi} \!\! \overline{\psi} (\theta,t) \dd\theta .  
\end{align}
By using Green Formula and periodic boundary conditions \eqref{eq:BC}, we obtain
\begin{equation*}
  \int_{-\pi}^{\pi}  \partial^2_{\theta\theta}  \psi (\theta,t) 
 \, \overline{\psi}(\theta,t) \dd\theta  = -   \int_{-\pi}^{\pi}  \left| \partial_\theta \psi (\theta,t) \right|^2   \dd\theta  .
\end{equation*}
Thus, considering the real part of \eqref{eq:1115}, we deduce that
\begin{equation}\label{eq:1232}
\frac{1}{2} \frac{\partial}{\partial t} \int_{-\pi}^{\pi}  |\psi(\theta,t)|^2\dd\theta  + \|\psi(t)\|^2 = F     \int_{-\pi}^{\pi} \!\! \Re({\psi}) (\theta,t) \dd\theta 
\end{equation}
and the Cauchy-Schwarz inequality yields
\begin{align*}
F     \int_{-\pi}^{\pi} \!\! \Re({\psi}) (\theta,t) \dd\theta 
&\leq
\pi\, F^2 + \frac{1}{2}   \int_{-\pi}^{\pi} \!\! \Re({\psi})^2 (\theta,t) \dd\theta  \\
&\leq \pi F^2 + \frac{1}{2} \,\|\psi(t)\|^2 .
\end{align*}
We deduce from \eqref{eq:1232} that $\frac{\partial}{\partial t}   \|\psi(t)\|^2  + \|\psi(t)\|^2 
\leq 2 \pi F^2 $
and it follows  that
\begin{equation*}
\frac{\partial}{\partial t} \big( \e^t  \|\psi(t)\|^2 \big) 
\leq 2   \pi F^2 \,  \e^t   .
\end{equation*} 
Finally, integrating this inequality for $t\in[0,T]$, we deduce the following energy estimate for the solution to the FP-LLE:
\begin{equation}\label{eq:1124a}
\forall T>0\quad
\|\psi(T)\|^2\leq   \e^{-T} \, \|\psi_0\|^2+  2\pi F^2\,(1-\e^{-T}) . 
\end{equation}

%====================
\section{Singular points on the curve of flat solutions}\label{app:1125}
%====================
The steady state FP-LLE can be written as   $\cF(\psi)=0$  where
\[
   \cF(\psi) = -i\frac{\beta}{2}\partial_{\theta\theta}\psi - (1+i\alpha)\psi
   +i\psi\Big(|\psi|^2 + \frac{\sigma}{\pi}\|\psi\|^2\Big) + F . 
\]
A necessary condition for a flat solution $\psi_\bullet$ to be the onset of a bifurcation is that $\partial_\psi\cF(\psi_\bullet)$  is not invertible.
The search for  bifurcation points can be notably simplified, a fact highlighted in \cite{Miyaji:10}, if instead of~$\cF$, one considers  the FP-LLE in the form
\[
   \cF_\bullet(v) \eqdef\psi_\bullet^{-1} \,\cF\big(\psi_\bullet(1+v)\big) = 0
\]
with  $v$ as new unknown.
Differentiating $\cF_\bullet$ at $0$ yields
\[
   \partial_v\cF_\bullet(0) = \psi_\bullet^{-1} \,\partial_\psi\cF(\psi_\bullet) \,\psi_\bullet\,.
\]
The complex valued operator $\partial_v\cF_\bullet(0)$ acts on complex valued functions $w$.
Considering the real and imaginary parts of $w$ and $\partial_v\cF_\bullet(0)[w]$, one transforms the expression of $\partial_v\cF_\bullet(0)$ into a $2\times2$ matrix of real valued operator blocks.
Considering the column vector $(\Im(\partial_v\cF_\bullet(0)[w]),-\Re(\partial_v\cF_\bullet(0)[w]))^\top$ as a function of $(\Re w,\Im w)^\top$, we find that the conditions for non-invertibility of
 $\partial_\psi\cF(\psi_\bullet)$
 are given by the non-invertibility of  
\begin{multline}
\label{eq:M+J+K}
    M+J_\bullet+K_\bullet = 
   \begin{pmatrix}
   \frac{\beta}{2}\cA -\alpha \Id & - \Id \\[1ex]
    \Id & \frac{\beta}{2}\cA  - \alpha \Id
   \end{pmatrix}
   \\
   +
   \begin{pmatrix}
   3\rho_\bullet \Id & 0\\[1ex]
   0 & \rho_\bullet \Id
   \end{pmatrix}
   + 
   \begin{pmatrix}
   2\sigma\rho_\bullet (\Id + \cI) & 0\\[1ex]
   0 & 2\sigma\rho_\bullet \Id
   \end{pmatrix}
\end{multline}
in which appear the operators $\cA:w\mapsto-\partial_{\theta\theta}w$ and $\cI:w\mapsto \frac{1}{\pi}\int_{-\pi}^\pi w$. Note that $M$ corresponds to the linear part $-i\frac{\beta}{2}\partial_{\theta\theta} - (1+i\alpha)\Id$ of $\cF$ while $J_\bullet$ and $K_\bullet$ come from the linearization of the cubic parts $\psi\mapsto i\psi|\psi|^2$ and $\psi\mapsto\frac{\sigma}{\pi}i\psi\|\psi\|^2$. Relying on the diagonalization of the operators $\cA$ and $\cI$ in the trigonometric basis $\cos k\theta$, $k\in\N$, and $\sin k\theta$, $k\in\N^*$, it is straightforward to deduce the block diagonalization $ (\widehat L_k)_{k\in\N}$ of $M+J_\bullet+K_\bullet$ where
{\small
\begin{equation}
\label{eq:Lk}
 \hspace*{-4mm}   \widehat L_k =  \begin{pmatrix}
   \frac{\beta}{2}\,k^2-\alpha + \rho_\bullet(3+2\sigma) + 4\sigma\rho_\bullet\delta_{0k}
   & - 1 \\[1ex]
    1  & \hspace*{-15mm} \frac{\beta}{2}\,k^2 - \alpha + \rho_\bullet(1 + 2\sigma)
   \end{pmatrix}
\end{equation}
}%
from which we find the following dispersion relations necessary to have a bifurcation point at $\psi_\bullet$  (compare with \cite[eq. (54)]{Cole:18}): $  \exists k\in\N,\ \exists\varepsilon\in\{\pm1\}$ such that
\begin{eqnarray}\label{eq:disp}
& \hspace*{-15mm}
   -\alpha + 2\rho_\bullet(1+2\sigma) = \varepsilon\sqrt{\rho_\bullet^2(1+2\sigma)^2-1}
    &\mbox{if}\quad k=0\\[1ex]
& \hspace*{-18mm}  \frac{\beta}{2}\,k^2-\alpha + 2\rho_\bullet(1+\sigma) = \varepsilon\sqrt{\rho_\bullet^2-1}
   & \mbox{if}\quad k\neq0\nonumber .
\end{eqnarray}

\medskip

A similar investigation of necessary conditions for bifurcations from the curve of flat solutions can be achieved for the discrete version
\eqref{eq:1218}  of the steady-state FP-LLE.
When $U=U_\bullet$ is associated to a flat solution $\psi_\bullet=u_{\bullet,1}+iu_{\bullet,2}$, we have $u_{1,n}=u_{\bullet,1}$ and $u_{2,n}=u_{\bullet,2}$ for all $n$, and $U^\top U = N(u_{\bullet,1}^2+u_{\bullet,2}^2)=N\rho_\bullet$. Hence formula~\eqref{eq:1639} becomes
{\small
\begin{multline*}
   \partial_UG_h(U_\bullet,\lambda) = \mathrm{M} +
   \begin{pmatrix}
   (3u_{\bullet,1}^2+u_{\bullet,2}^2) \,\mathrm{I}_N & 2u_{\bullet,1}u_{\bullet,2} \,\mathrm{I}_N \\[1ex]
   2u_{\bullet,1}u_{\bullet,2} \,\mathrm{I}_N  & (u_{\bullet,1}^2+3u_{\bullet,2}^2)  \,\mathrm{I}_N 
   \end{pmatrix} 
   \\
  + 2\sigma\rho_\bullet  \,\mathrm{I}_{2N} + 
   \frac{2\sigma}{N} \begin{pmatrix}
   u_{\bullet,1}^2  \,\mathrm{K}_N & u_{\bullet,1}u_{\bullet,2}\, \,\mathrm{K} _N\\[1ex]
   u_{\bullet,1}u_{\bullet,2}\, \,\mathrm{K} _N & u_{\bullet,2}^2  \,\mathrm{K} _N
   \end{pmatrix} 
\end{multline*}
}%
where $\mathrm{M}$ is  given by \eqref{eq:1010}, $\mathrm{K}_N$ denotes the $N\times N$ matrix with all  coefficients  $1$. 

By a pre-multiplication by the matrix $\mathrm{P} ^{-1}$    and a post-multiplication by the matrix  $\mathrm{P}$  where 
\[
\mathrm{P}  = \begin{pmatrix}
u_{\bullet,1}\,  \mathrm{I}_N & -u_{\bullet,2}\,  \mathrm{I}_N \\
u_{\bullet,2}\,  \mathrm{I}_N & u_{\bullet,1}\,  \mathrm{I}_N
\end{pmatrix}
\]
we obtain that  $  \partial_UG_h(U_\bullet,\lambda) $  is similar to the matrix
\begin{eqnarray}
\label{eq:MN+JN+KN}
   \mathrm{M}+ \rho_\bullet \begin{pmatrix}
   3  \,\mathrm{I}_N & 0 \\ 0 &\mathrm{I}_N
\end{pmatrix} + 
\rho_\bullet\begin{pmatrix}
   2\sigma \,\mathrm{I}_N & 0 \\ 0 & 2\sigma  \,\mathrm{I}_N
\end{pmatrix} 
\nonumber\\+
   \frac{2\sigma}{N} \begin{pmatrix}
   \rho_\bullet\, \mathrm{K}_N & 0\\[1ex]
   0 & 0
   \end{pmatrix} \,.
\end{eqnarray}
Note that  \eqref{eq:MN+JN+KN} is the finite difference discretization of~\eqref{eq:M+J+K}, combined with trapezoidal quadrature formula.

The matrix $\mathrm{A}$ given by \eqref{eq:1025} can be easily  diagonalized  by noting that it is a circulant matrix with non-zero coefficients $c_0=2$, $c_1=-1$, and $c_{N-1}=-1$.
We  deduce that its eigenvalues are $\lambda_k=\sum_{j=0}^{N-1}c_j e^{ij\frac{2k\pi}{N}}= 2(1-\cos\frac{2k\pi}{N})$ for $k=0,\ldots,N-1$, associated with eigenvectors $V_k=(1,e^{i\frac{2k\pi}{N}},\ldots,e^{i(N-1)\frac{2k\pi}{N}})^\top$. Noticing that $\lambda_k=\lambda_{N-k}$ for $k=1,\ldots,[\frac{N}{2}]$ and $V_{N-k}=\overline V_k$, we find the basis of real eigenvectors
$
   \Big(1,\cos\tfrac{2k\pi}{N},\ldots,\cos (N-1)\tfrac{2k\pi}{N}\Big)^\top
   $ 
   and
   $
   \Big(0,\sin\tfrac{2k\pi}{N},\ldots,\sin (N-1)\tfrac{2k\pi}{N}\Big)^\top
$
that correspond to  the eigenvectors $\cos k\theta$ and $\sin k\theta$ of $\cA$
evaluated at the grid nodes $\theta_n, n=0,\ldots, N-1$.

With this information at hand, we find that the matrix given in \eqref{eq:MN+JN+KN} is similar to the diagonal matrix with diagonal block given by  $(\widehat L_{N,k})_{k=0,\ldots,[\frac{N}{2}]}$ where
{\small
\begin{equation*}
\label{eq:LNk}
   \widehat L_{N,k} =  \begin{pmatrix}
   \frac{\beta}{2}\,\omega_{N,k} -\alpha 
   + \rho_\bullet(3+2\sigma) + 4\sigma\rho_\bullet\delta_{0k} 
   & - 1 \\[1ex]
    1 & \hspace*{-15mm}  \frac{\beta}{2}\,\omega_{N,k}  - \alpha + \rho_\bullet(1 + 2\sigma)
   \end{pmatrix}
\end{equation*}
}%
and $\omega_{N,k}:=\frac{2}{h^2}\,(1-\cos\frac{2k\pi}{N}) = \frac{2}{h^2}\,(1-\cos kh)$. 
We deduce that necessary conditions for a bifurcation point at $U_\bullet$ are given by a modification of \eqref{eq:disp} in which $k^2$ is replaced by $\omega_{N,k}$ and $k\le\frac{N}{2}$.
Note that when the step-size $h$ tends to $0$,  we retrieve from the discrete dispersion relation the continuous dispersion relation~\eqref{eq:disp}.

%=======================================
\bibliographystyle{siamplain}
\bibliography{biblio}

\end{document}